\documentclass[conference]{IEEEtran}

\usepackage[utf8]{inputenc}
\usepackage[english]{babel}
\usepackage{cite}
\usepackage{amsmath,amssymb,amsfonts}
\usepackage{subfig}    
\usepackage[demo]{graphicx}
\usepackage{textcomp}
\usepackage{xcolor}
\usepackage{mathrsfs}
\usepackage[utf8]{inputenc}
\usepackage{comment}
\usepackage[ruled,vlined]{algorithm2e}
\usepackage{mathtools}
\usepackage[font=small,labelfont=bf]{caption}
\usepackage{blkarray, bigstrut}
\usepackage{subfig}
\usepackage{cases}

\usepackage{enumerate}

\SetCommentSty{mycommfont}
\SetKwInput{KwInput}{Input}                
\SetKwInput{KwOutput}{Output}              

\usepackage{enumitem}
\usepackage{float}

\usepackage{pgf, tikz}
\usetikzlibrary{arrows, automata}
\usepackage{tikz}
\usetikzlibrary{automata,positioning}

\newcommand*\circled[2][1.6]{\tikz[baseline=(char.base)]{
    \node[shape=circle, draw, inner sep=1pt, 
        minimum height={\f@size*#1},] (char) {\vphantom{WAH1g}#2};}}

\newlist{legal}{enumerate}{10}
\setlist[legal]{label*=\arabic*.}
\usepackage{ifpdf}

\usepackage{amsthm}

\newtheoremstyle{mystyle1}
  {.5pt}
  {.5pt}
  {}
  {}
  {\bfseries}
  {.}
  { }
  {\thmname{#1}\thmnumber{ #2}\thmnote{ (#3)}}

\theoremstyle{mystyle1}

\newtheorem{defn}{Definition}

\newtheorem{rem}{Remark}
\newtheorem{exmp}{Example}

\newtheoremstyle{mystyle}
  {.5pt}
  {.5pt}
  {\itshape}
  {}
  {\bfseries}
  {.}
  { }
  {\thmname{#1}\thmnumber{ #2}\thmnote{ (#3)}}

 \theoremstyle{mystyle}

\newtheorem{thm}{Theorem}
\newtheorem{lem}{Lemma}
\newtheorem{prop}{Proposition}

\usepackage[utf8]{inputenc}
\usepackage[english]{babel}
\usepackage{cite}
\usepackage{amsmath,amssymb,amsfonts}
\usepackage{algorithmic}
\usepackage{graphicx}
\usepackage{textcomp}
\usepackage{xcolor}
\usepackage{mathrsfs} 
\usepackage{enumitem}
\usepackage{amsthm}
\usepackage{cite}
\usepackage{array}
\usepackage{amsmath}

\setlist[legal]{label*=\arabic*.}

\theoremstyle{mystyle1}

\usepackage{enumitem}
\usepackage{geometry}
\allowdisplaybreaks

\begin{document}

\title{Generalization of the Fano and Non-Fano Index Coding Instances}

\author{\IEEEauthorblockN{Arman Sharififar, Parastoo Sadeghi, Neda Aboutorab} 
}

\maketitle

\begin{abstract}
Matroid theory is fundamentally connected with index coding and network coding problems. In fact, the reliance of linear index coding and network coding rates on the characteristic of a field has been demonstrated by using the two well-known matroid instances, namely the Fano and non-Fano matroids. This established the insufficiency of linear coding, one of the fundamental theorems in both index coding and network coding. While the Fano matroid is linearly representable only over fields with characteristic two, the non-Fano instance is linearly representable only over fields with odd characteristic. 
For fields with arbitrary characteristic $p$,  the Fano and non-Fano matroids were extended to new classes of matroid instances whose linear representations are dependent on fields with characteristic $p$. However, these matroids have not been well appreciated nor cited in the fields of network coding and index coding. In this paper, we first reintroduce these matroids in a more structured way. Then, we provide a completely independent alternative proof with the main advantage of using only matrix manipulation rather than complex concepts in number theory and matroid theory. In this new proof, it is shown that while the class $p$-Fano matroid instances are linearly representable only over fields with characteristic $p$, the class $p$-non-Fano instances are representable over fields with any characteristic other than characteristic $p$.
Finally, following the properties of the class $p$-Fano and $p$-non-Fano matroid instances, we characterize two new classes of index coding
instances, respectively, referred to as the class $p$-Fano and $p$-non-Fano index coding, each with a size of $p^2 + 4p + 3$. It is proved that the
broadcast rate of the class $p$-Fano index coding instances
is achievable by a linear code over only fields with characteristic $p$. In contrast, it is shown that for the class $p$-non-Fano index coding instances, the broadcast rate is achievable by a linear code over fields with any characteristic other than characteristic $p$. 

\end{abstract}

\section{Introduction}
Founded by Whitney \cite{Whitney}, matroid theory is a discipline of mathematics in which the notions of dependence and independence relations are generalized. These concepts in matroid theory are closely related to the concepts in information theory such as entropy and the notions in linear algebraic coding theory including the dependent and independent linear vector spaces. Among the fields in coding theory, network coding and index coding are proven to be fundamentally connected to the matroid theory \cite{Dougherty2007,Rouayheb-matroid,Dougherty-matroid,Yin-network-matroid,Bassoli-survey,Sun-matroid,Rouayheb2010,Thomas,Arman-4-3}.

Network coding, started by the work of Ahlswede \textit{et al.} \cite{Ahlswede}, studies the communication problem in which a set of messages must be communicated from source nodes to destination nodes through intermediate nodes. Encoding messages at the intermediate nodes can increase the overall network throughput compared to simply routing the messages \cite{Ahlswede}. The basic connection between matroid theory and network coding was established in \cite{Dougherty2007}, where a mapping technique from any matroid instance to a network coding instance was presented. It was shown that the constraints on the column space of the matrix which is a linear representation of the matroid instance can equivalently be mapped to the column space of the global and local encoding matrices solving the corresponding network coding instance. 

Index coding, introduced by the work of Birk and Kol \cite{Birk1998,Birk2006}, models a broadcast communication system in which a single server must communicate a set of messages to a number of users. Each user demands specific messages and may have prior knowledge about other messages. Encoding the messages at the server can improve the overall broadcast rate rather than sending uncoded messages \cite{Birk2006}. In \cite{Rouayheb2010}, matroid theory was shown to be closely bonded with index coding by providing a systematic approach to reduce a matroid instance into an index coding instance. This method maps the constraints on the column space of the matrix which linearly represents a matroid instance to the column space of the encoding matrix solving the corresponding index coding instance. 

Due to the strong connection of matroid theory with network coding and index coding problems, a number of well-known matroid instances were used to prove some key limitations in both network and index coding. As an example, by mapping the non-Pappus matroid instance which does not have a scalar linear representation into its corresponding network and index coding instances, it can be shown that vector linear codes can outperform the optimal scalar linear coding in both network and index coding \cite{Rouayheb-matroid, Rouayheb2010}. Moreover, through the Vamos matroid instance which is not linearly representable, the necessity of non-Shannon inequalities was proved for both network coding and index coding \cite{Dougherty2007, Sun-non-shannon-exp-1} to establish a tighter performance bound. 
In addition, the insufficiency of linear coding was established in \cite{Dougherty2005, Rouayheb2010} by providing two network and index coding instances which are essentially related to the Fano and non-Fano matroid instances. The dependency of linear representation of the Fano and non-Fano matroids on the characteristic of a field is the key reason for the insufficiency of linear coding. 

In \cite{Bernt-matroid}, for fields with arbitrary characteristic $p$, the Fano and non-Fano matroid instances were extended to the new classes of matroids, respectively, referred to as the class $p$-Fano and the class $p$-non-Fano matroids. Using the concept of matroid theory and number theory, it was proved that while the class $p$-Fano instances are linearly representable only over fields with characteristic $p$, the class $p$-non-Fano instances have linear representation over fields with any characteristic other than characteristic $p$ \cite{Bernt-matroid, pena-matroid}. These matroids are of significant importance to the fields of network coding and index coding as they can be mapped into new classes of network and index coding instances whose optimal linear coding rates are dependent on fields with characteristic $p$. However, to the best of our knowledge, these classes of matroid instances have not been studied in the community of network coding and index coding.


In this paper, we first reintroduce the class $p$-Fano and the class $p$-non-Fano matroid instances in a more structured way. Then we provide a completely independent alternative proof for their linear representations' dependency on fields with characteristic $p$. The main advantage of our proof (proof of Theorems \ref{thm:p-Fano-matroid} and \ref{thm:p-non-Fano-matroid}) is that it uses only matrix manipulation in linear algebra rather than complex concepts in matroid theory and number theory. 
Finally, following the class $p$-Fano and $p$-non-Fano matroid instances, we characterize two new classes of index coding
instances, respectively, referred to as the class $p$-Fano and $p$-non-Fano index coding, each with a size of $2p^2 + 4p + 3$. For the class $p$-Fano index coding instances, it is proved that linear codes are optimal only over fields with characteristic $p$. However, for the class $p$-non-Fano index coding instances, we prove that linear codes are optimal over fields with any characteristic other than characteristic $p$. It is important to note that using the mapping methods discussed in \cite{Rouayheb2010} and \cite{Maleki2014} for the class $p$-Fano and $p$-non-Fano matroids results in index coding instances of size $\mathcal{O}(2^p)$, while the instances of $p$-Fano and $p$-non-Fano matroids presented in this paper have a size of $\mathcal{O}(p^2)$.

The organization and contributions of this paper are summarized as follows.
\begin{itemize}
    \item Section \ref{sec:introduction} provides a brief overview of the system model, relevant background and definitions in matroid theory, index coding.
    \item In Section \ref{sec:p-Fano-non-Fano}, for the two recognized categories of matroids, namely the class $p$-Fano and the class $p$-non-Fano matroid instances,
    we introduce a fully distinct alternative proof that exclusively employs matrix manipulation techniques. This proof demonstrates that instances of the class $p$-Fano matroid can only be linearly represented over fields with characteristic p, whereas instances of the class $p$-non-Fano matroid can be linearly represented over fields with any characteristic, except for characteristic $p$.
    \item In Subsection \ref{sub:index-coding-fano}, we characterize a new class of index coding, referred to as the class $p$-Fano index coding instances with a size of $2p^{2}+4p+3$. It is proved that the broadcast rate of the class $p$-Fano index coding instances is achievable by linear code over only fields with characteristic $p$.
    \item In Subsection \ref{sub:index-coding-non-fano}, we characterize a new class of index coding, referred to as the class $p$-non-Fano index coding instances with a size of $2p^{2}+4p+3$. It is proved that the broadcast rate of the class $p$-non-Fano index coding instances is achievable by linear code over fields with any characteristic other than characteristic $p$.
\end{itemize}


   

\section{SYSTEM MODEL AND BACKGROUND} \label{sec:introduction}

\subsection{Notation}
Small letters such as $n$  denote an integer where  $[n]\triangleq\{1,...,n\}$ and $[n:m]\triangleq\{n,n+1,\dots m\}$ for $n\leq m$. Capital letters such as $L$ denote a set, with $|L|$ denoting its cardinality. Symbols in boldface such as $\boldsymbol{l}$ and $\boldsymbol{L}$, respectively, denote a vector and a matrix, with $\mathrm{rank}(\boldsymbol{L})$ and $\mathrm{col}(\boldsymbol{L})$ denoting the rank and column space of matrix $\boldsymbol{L}$, respectively. A calligraphic symbol such as $\mathcal{L}$ denotes a set whose elements are sets.\\
We use $\mathbb{F}_q$ to denote a finite field of size $q$ and write $\mathbb{F}_{q}^{n\times m}$ to denote the vector space of all $n\times m$ matrices over the field $\mathbb{F}_{q}$.
$\boldsymbol{I}_{n}$ denotes the identity matrix of size $n\times n$, and $\boldsymbol{0}_{n}$ represents an $n\times n$ matrix whose elements are all zero.

\subsection{System Model} \label{sub:background-system-model}
Consider a broadcast communication system in which a server transmits a set of $mt$ messages $X=\{x_{i}^{j},\ i\in[m],\ j\in [t]\},\ x_{i}^{j}\in \mathcal{X}$, to a number of users $U=\{u_i,\ i\in[m]\}$ through a noiseless broadcast channel. Each user $u_i$ wishes to receive a message of length $t$, $X_i=\{x_{i}^{j},\  j\in[t]\}$ and may have a priori knowledge of a subset of the messages $S_i:=\{x_{l}^{j},\ l\in A_{i},\ j\in[t]\},\ A_{i}\subseteq[m]\backslash \{i\}$, which is referred to as its side information set. The main objective is to minimize the number of coded messages which is required to be broadcast so as to enable each user to decode its requested message.
An instance of index coding problem $\mathcal{I}$ can be either characterized by the side information set of its users as $\mathcal{I}=\{A_{i}, i\in[m]\}$, or by their interfering message set $B_{i}=[m]\backslash (A_i \cup \{i\})$ as $\mathcal{I}=\{B_{i}, i\in[m]\}$. 

\begin{defn}[$\mathcal{C}_{\mathcal{I}}$: Index Code for $\mathcal{I}$]
Given an instance of index coding problem $\mathcal{I}=\{A_{i}, i\in[m]\}$, a $(t,r)$ index code is defined as $\mathcal{C}_{\mathcal{I}}=(\phi_{\mathcal{I}},\{\psi_{\mathcal{I}}^{i}\})$, where
 \begin{itemize}
     \item $\phi_{\mathcal{I}}: \mathcal{X}^{mt}\rightarrow \mathcal{X}^{r}$ is the encoding function which maps the $mt$ message symbol $x_{i}^{j}\in \mathcal{X}$ to the $r$ coded messages as $Y=\{y_{1},\dots,y_{r}\}$, where $y_k\in \mathcal{X}, \forall k\in [r]$.
     \item $\psi_{\mathcal{I}}^{i}:$ represents the decoding function, where for each user $u_i, i\in[m]$, the decoder $\psi_{\mathcal{I}}^{i}: \mathcal{X}^{r}\times \mathcal{X}^{|A_i|t}\rightarrow \mathcal{X}^{t}$ maps the received $r$ coded messages $y_k\in Y, k\in[r]$ and the $|A_i|t$ messages $x_{l}^{j}\in S_i$ in the side information to the $t$ messages $\psi_{\mathcal{I}}^{i}(Y,S_i)=\{\hat{x}_{i}^{j}, j\in [t]\}$, where $\hat{x}_{i}^{j}$ is an estimate of $x_{i}^{j}$.
 \end{itemize}
\end{defn}

\begin{defn}[$\beta(\mathcal{C}_{\mathcal{I}})$: Broadcast Rate of $\mathcal{C}_{\mathcal{I}}$]
Given an instance of the index coding problem $\mathcal{I}$, the broadcast rate of a $(t,r)$ index code $\mathcal{C}_{\mathcal{I}}$ is defined as $\beta(\mathcal{C}_{\mathcal{I}})=\frac{r}{t}$.
\end{defn}

\begin{defn}[$\beta(\mathcal{I})$: Broadcast Rate of $\mathcal{I}$]
Given an instance of the index coding problem $\mathcal{I}$, the broadcast rate $\beta(\mathcal{I})$ is defined as
\begin{equation} \label{eq:opt-rate}
    \beta(\mathcal{I})=\inf_{t} \inf_{\mathcal{C}_{\mathcal{I}}} \beta(\mathcal{C}_{\mathcal{I}}).
\end{equation}
Thus, the broadcast rate of any index code $\mathcal{C}_{\mathcal{I}}$ provides an upper bound on the broadcast rate of $\mathcal{I}$, i.e., $\beta(\mathcal{I}) \leq \beta(\mathcal{C}_{\mathcal{I}})$.
\end{defn}

\subsection{Linear Index Code} \label{sub:background-linear-index-code}
Let $\boldsymbol{x}=[\boldsymbol{x}_{1},\dots,\boldsymbol{x}_{m}]^{T}\in \mathbb{F}_{q}^{mt\times 1}$ denote the message vector.

\begin{defn}[Linear Index Code]
Given an instance of the index coding problem $\mathcal{I}=\{B_{i}, i\in[m]\}$, a $(t,r)$ linear index code is defined as $\mathcal{L}_{\mathcal{I}}=(\boldsymbol{H},\{\psi_{\mathcal{I}}^{i}\})$, where
  \begin{itemize}
      \item $\boldsymbol{H}: \mathbb{F}_{q}^{mt\times 1}\rightarrow \mathbb{F}_{q}^{r\times 1}$ is the $r\times mt$ encoding matrix which maps the message vector $\boldsymbol{x}\in \mathbb{F}_{q}^{mt\times 1}$  to a coded message vector $\boldsymbol{\Bar{y}}=[y_{1},\dots,y_{r}]^{T}\in \mathbb{F}_{q}^{{r}\times 1}$ as follows
      \begin{equation} \nonumber
      \boldsymbol{y}=\boldsymbol{H}\boldsymbol{x}=\sum_{i\in [m]} \boldsymbol{H}^{\{i\}}\boldsymbol{x}_i.
      \end{equation}
      Here $\boldsymbol{H}^{\{i\}}\in \mathbb{F}_{q}^{r\times t}$ is the local encoding matrix of the $i$-th message $\boldsymbol{x}_i$ such that
      $\boldsymbol{H}=
      \left [\begin{array}{c|c|c}
        \boldsymbol{H}^{\{1\}} & \dots & \boldsymbol{H}^{\{m\}}
      \end{array}
     \right ]\in \mathbb{F}_{q}^{r\times mt}$.
     \item $\psi_{\mathcal{I}}^{i}$ represents the linear decoding function for user $u_{i}, i\in[m]$, where $\psi_{\mathcal{I}}^{i}(\boldsymbol{y}, S_i)$ maps the received coded message $\boldsymbol{y}$ and its side information messages $S_i$ to $\hat{\boldsymbol{x}}_{i}$, which is an estimate of the requested message vector $\boldsymbol{x}_i$.
  \end{itemize}
\end{defn}

\begin{prop} [\cite{Arman-3-2}]\label{prop-lineardecoding-condition}
The necessary and sufficient condition for linear decoder $\psi_{\mathcal{I}}^{i}, \forall i\in[m]$ to correctly decode the requested message vector $\boldsymbol{x}_i$ is
   \begin{equation} \label{eq:dec-cond}
       \mathrm{rank} (\boldsymbol{H}^{\{i\}\cup B_{i}})= \mathrm{rank} (\boldsymbol{H}^{B_{i}}) + t,
   \end{equation}
where $\boldsymbol{H}^L$ denotes the matrix $\left [\begin{array}{c|c|c}
       \boldsymbol{H}^{\{l_1\}} & \dots & \boldsymbol{H}^{\{l_{|L|}\}}
     \end{array}
    \right ]$ for the given set $L=\{l_1,\dots,l_{|L|}\}$. 
\end{prop}

\begin{defn}[$\lambda_{q}(\mathcal{L}_{\mathcal{I}})$: Linear Broadcast Rate of $\mathcal{L}_{\mathcal{I}}$ over $\mathbb{F}_q$]
Given an instance of index coding problem $\mathcal{I}$, the linear broadcast rate of a $(t,r)$ linear index code $\mathcal{L}_{\mathcal{I}}$ over field $\mathbb{F}_q$ is defined as $\lambda_{q}(\mathcal{L}_{\mathcal{I}})=\frac{r}{t}$.
\end{defn}

\begin{defn}[$\lambda_{q}(\mathcal{I})$: Linear Broadcast Rate of $\mathcal{I}$ over $\mathbb{F}_q$]
Given an instance of index coding problem $\mathcal{I}$, the linear broadcast rate $\lambda_{q}(\mathcal{I})$ over field $\mathbb{F}_q$ is defined as
    \begin{equation} \nonumber
        \lambda_{q}(\mathcal{I})=\inf_{t} \inf_{\mathcal{L}_{\mathcal{I}}} \lambda_{q}(\mathcal{L}_{\mathcal{I}}).
    \end{equation}
\end{defn}

\begin{defn}[$\lambda(\mathcal{I})$: Linear Broadcast Rate for $\mathcal{I}$]
Given an instance of index coding problem $\mathcal{I}$, the linear broadcast rate is defined as 
\begin{equation} \label{eq:opt-lin-rate}
    \lambda(\mathcal{I})=\min_{q} \lambda_{q}(\mathcal{I}).
\end{equation}
\end{defn}

\begin{defn}[Scalar and Vector Linear Index Code]
The linear index code $\mathcal{C}_{\mathcal{I}}$ is said to be scalar if $t=1$. Otherwise, it is called a vector (or fractional) code. For scalar codes, we use $x_{i}=x_{i}^{1}, i\in [m]$, for simplicity.
\end{defn}

\subsection{Graph Definitions} \label{sub:background-graph}
Given an index coding instance $\mathcal{I}$, the following concepts are defined based on its interfering message sets, which are, in fact, related to its graph representation \cite{Arman-4-1, Arman-5, Liu-secure-7, Liu-secure-6}.

\begin{defn}[Independent Set of $\mathcal{I}$]
We say that set $M\subseteq[m]$ is an independent set of $\mathcal{I}$ if $B_{i}\cap M=M\backslash \{i\}$ for all $i\in M$.
\end{defn}

\begin{defn}[Minimal Cyclic Set of $\mathcal{I}$] \label{def:minimal-cyclic-set}
Let $M=\{i_{j}, j\in [|M|]\}\subseteq[m]$. Now, $M$ is referred to as a minimal cyclic set of $\mathcal{I}$ if
 \begin{equation} \label{eq:def-minimal-cyclic-set}
 B_{i_{j}}\cap M=
    \left\{
      \begin{array}{cc}
         M\backslash \{i_{j}, i_{j+1}\},\ \ \ \ \ \ \ \ j\in [|M|-1],
         \\ \\
        M\backslash \{i_{|M|}, i_{1}\},\ \ \ \ \ \ \ \ \ \ \ \  j=i_{|M|}.\ \ \
     \end{array}
     \right.
\end{equation}
\end{defn}

\begin{defn}[Acyclic Set of $\mathcal{I}$] \label{def:acyclic-set}
We say that $M\subseteq[m]$ is an acyclic set of $\mathcal{I}$, if none of its subsets $M^{\prime}\subseteq M$ forms a minimal cyclic set of $\mathcal{I}$. We note that each independent set is an acyclic set as well.
\end{defn}

\begin{prop}[\cite{Arbabjolfaei2018}] \label{prop:rate-acyclic-minimal}
Let $\mathcal{I}=\{B_{i}, i\in [m]\}$. It can be shown that
\begin{itemize}[leftmargin=*]
    \item if set $[m]$ is an acyclic set of $\mathcal{I}$, then $\lambda_{q}(\mathcal{I})=\beta(\mathcal{I})=m$.
    \item if set $[m]$ is a minimal cyclic set of $\mathcal{I}$, then $\lambda_{q}(\mathcal{I})=\beta(\mathcal{I})=m-1$.
\end{itemize}

\end{prop}

\begin{defn}[Maximum Acyclic Induced Subgraph (MAIS) of $\mathcal{I}$]
Let $\mathcal{M}$ be the set of all sets $M\subseteq[m]$ which are acyclic sets of $\mathcal{I}$. Then, set $M\in \mathcal{M}$ with the maximum size $|M|$ is referred to as the MAIS set of $\mathcal{I}$, and $\beta_{\text{MAIS}}(\mathcal{I})=|M|$ is called the MAIS bound for $\lambda_{q}(\mathcal{I})$, as we always have \cite{Bar-Yossef2011}
\begin{align}
    \lambda_{q}(\mathcal{I})\geq \beta_{\text{MAIS}}(\mathcal{I}).
    \label{eq:MAIS-linear-coding}
\end{align}
\end{defn}

\begin{rem}
Equation \eqref{eq:MAIS-linear-coding} establishes a sufficient condition for optimality of the linear coding rate as follows.
Given an index coding instance $\mathcal{I}$, if $\lambda_{q}(\mathcal{I})= \beta_{\text{MAIS}}(\mathcal{I})$, then the linear coding rate is optimal for $\mathcal{I}$. In this paper, the encoding matrix which achieves this optimal rate is denoted by $\boldsymbol{H}_{\ast}$.
\end{rem}


\subsection{Overview of Matroid Theory} \label{sub:background-matroid}

\begin{defn}[$\mathcal{N}$: Matroid Instance \cite{Rouayheb2010,Thomas}] \label{def:matroid-general}
A matroid instance $\mathcal{N}=\{f(N), N\subseteq [n]\}$ is a set of functions $f: 2^{[n]}\rightarrow \{0,1,2,\dots \}$ that satisfy the following three conditions:
\begin{align}
    &f(N) \leq |N|,\ \ \ \ \ \ \ \ \ \ \ \ \ \ \ \ \ \ \ \ \ \ \ \ \ \ \ \ \ \ \ \ \ \ \ \ \forall N\subseteq[n],
    \nonumber
    \\
    &f(N_{1}) \leq f(N_{2}),\ \ \ \ \ \ \ \ \ \ \ \ \ \ \ \ \ \ \ \ \ \ \ \ \ \ \ \ \ \ \ \ \forall N_{1}\subseteq N_{2}\subseteq[n],
    \nonumber
    \\
    &f(N_{1}\cup N_{2})+f(N_{1}\cap N_{2}) \leq f(N_{1})+f(N_{2}),\forall N_{1}, N_{2}\subseteq[n].
    \nonumber
\end{align}
Here, set $[n]$ and function $f(\cdot)$, respectively, are called the ground set and the rank function of $\mathcal{N}$. 
The rank of matroid $\mathcal{N}$ is defined as $f(\mathcal{N})=f([n])$.
\end{defn}

\begin{defn}[Basis and Circuit Sets of $\mathcal{N}$] \label{def:Basis-Circuit-Matroid}
Consider a matroid $\mathcal{N}$ of rank $f(\mathcal{N})$. We say that $N\subseteq[n]$ is an independent set of $\mathcal{N}$ if $f(N)=|N|$. Otherwise, $N$ is said to be a dependent set. A maximal independent set $N$ is referred to as a basis set. A minimal dependent set $N$ is referred to as a circuit set. 
\end{defn}
Let sets $\mathcal{B}$ and $\mathcal{C}$, respectively, denote the set of all basis and circuit sets of $\mathcal{N}$. Then, it can be shown that
\begin{align}
    &f(\mathcal{N})=f(N)=|N|,\ \ \ \ \ \ \ \ \ \ \ \ \forall N\in \mathcal{B},
    \nonumber
    \\
    &f(N\backslash \{i\})=|N|-1, \ \ \ \ \ \ \ \ \ \ \ \ \  \forall i\in N,\ \forall N\in \mathcal{C}. \label{eq:cicuitset}
\end{align}

\begin{defn}[$(t)$-linear Representation of $\mathcal{N}$ over $\mathbb{F}_{q}$] \label{def:linear-rep}
We say that matroid $\mathcal{N}=\{f(N), N\subseteq [n]\}$ of rank $f(\mathcal{N})$ has a $(t)$-linear representation over $\mathbb{F}_{q}$ if there exists a matrix 
\begin{align}
    \boldsymbol{H}=
      \left [\begin{array}{c|c|c}
        \boldsymbol{H}^{\{1\}} & \dots & \boldsymbol{H}^{\{n\}}
      \end{array}
     \right ],\ \  \boldsymbol{H}^{\{i\}}\in \mathbb{F}_{q}^{f(\mathcal{N})t\times t}, \forall i\in [n],
     \nonumber
\end{align}
such that
\begin{equation} \label{matroid-linear-representation}
    \mathrm{rank}(\boldsymbol{H}^{N})=f(N)t,\ \ \  \ \forall N\subseteq [n],
\end{equation}
where $\boldsymbol{H}^{N}$ denotes the matrix $\left [\begin{array}{c|c|c}
       \boldsymbol{H}^{\{n_1\}} & \dots & \boldsymbol{H}^{\{n_{|N|}\}}
     \end{array}
    \right ]$ for the given set $N=\{n_1,\dots,n_{|N|}\}\subseteq[n]$. 
\end{defn}
\begin{rem} \label{rem:1}
Since the matroid matrix $\boldsymbol{H}$ is related to the encoding matrix in index coding and network coding, we assume that each submatrix $\boldsymbol{H}^{\{i\}}, i\in [n]$ is invertible (which is a necessary condition for encoding matrix \cite{Arman-3-1,Arman-3-2,Arman-4-1}). 
\end{rem}
Now, based on Definitions \ref{def:Basis-Circuit-Matroid} and \ref{def:linear-rep}, the concepts of basis and circuit sets can also be defined for matrix $\boldsymbol{H}$.

\begin{defn}[Basis and Circuit Sets of $\boldsymbol{H}$] \label{def:Basis-Circuit-Matrix}
Let $N\subseteq [n]$. We say that $N$ is an independent set of $\boldsymbol{H}$, if $\mathrm{rank}(\boldsymbol{H}^{N})=|N|t$, otherwise $N$ is a dependent set of $\boldsymbol{H}$. The independent set $N$ is a basis set of $\boldsymbol{H}$ if $\mathrm{rank}(\boldsymbol{H})=\mathrm{rank}(\boldsymbol{H}^{N})=|N|t$. The dependent set $N$ is a circuit set of $\boldsymbol{H}$ if
\begin{align}
    \mathrm{rank}(\boldsymbol{H}^{N\backslash \{j\}})=\mathrm{rank}(\boldsymbol{H}^{N})=(|N|-1)t, \ \ \forall j\in N,
    \label{eq: def:circuit-set-H}
\end{align}
which requires that
\begin{align} 
\boldsymbol{H}^{\{j\}}=\sum_{i\in N\backslash \{j\}} \boldsymbol{H}^{\{i\}}\boldsymbol{M}_{j,i}, 
\label{eq:circuit-set-M-invertible}
\end{align}
where each $\boldsymbol{M}_{j,i}$ is invertible.
\end{defn}

\begin{defn}[Scalar and Vector Linear Representation]
If matroid $\mathcal{N}$ has a linear representation with $t=1$, it is said that $\mathcal{N}$ has a scalar linear representation. Otherwise, the linear representation is called a vector representation.
\end{defn}



\subsection{The Fano and non-Fano Matroid Instances $\mathcal{N}_{\text{F}}$ and $ \mathcal{N}_{\text{nF}}$} \label{sub:sub:Fano-nonFano}
\begin{defn}[Fano Matroid Instance $\mathcal{N}_{\text{F}}$ \cite{Dougherty2007}] \label{exm:Fano}
Consider the matroid instance $\mathcal{N}_{\text{F}}=\{(N,f(N)), N\subseteq[n]\}$ with $n=7$ and $f(\mathcal{N}_{\text{F}})=3$. Now, matroid $\mathcal{N}_{\text{F}}$ is referred to as the Fano matroid instance if set $N_{0}=[3]$ is a basis set, and the following sets $N_{i}, i\in [7]$ are circuit sets.
\begin{align}
    N_{1}&=\{1,2,4\},  
    N_{2} =\{1,3,5\}, 
    N_{3} =\{2,3,6\}, 
    N_{4} =\{1,6,7\},
    \nonumber
    \\
    N_{5}&=\{2,5,7\}, 
    N_{6} =\{3,4,7\}, 
    N_{7} =\{4,5,6\}.
    \label{eq:sets-Fano}
\end{align}
\end{defn}
\begin{prop}[\hspace{-0.05 ex}\cite{Dougherty2007}]  \label{prop-first-matroid-N_1}
The Fano matroid instance $\mathcal{N}_{\text{F}}$ is linearly representable over field $\mathbb{F}_{q}$ iff (if and only if) field $\mathbb{F}_{q}$ does have characteristic two.
\end{prop}
\begin{defn}[Non-Fano Matroid Instance $\mathcal{N}_{\text{nF}}$ \cite{Dougherty2007}] \label{exm:nonFano}
Consider the matroid instance $\mathcal{N}_{\text{nF}}=\{(N,f(N)), N\subseteq[n]\}$ with $n=7$ and $f(\mathcal{N}_{\text{nF}})=3$. Now, matroid $\mathcal{N}_{\text{nF}}$ is referred to as the non-Fano matroid instance if each set $N_{0}=[3]$ and $N_{7}=\{4,5,6\}$ is a basis set, and sets $N_{i}, i\in [6]$ in \eqref{eq:sets-Fano} are all circuit sets.
\end{defn}
\begin{prop}[\hspace{-0.05 ex}\cite{Dougherty2007}]  \label{prop-first-matroid-N_2}
The non-Fano matroid instance $\mathcal{N}_{\text{nF}}$ is linearly representable over $\mathbb{F}_{q}$ iff field $\mathbb{F}_{q}$ has odd characteristic (i.e., any characteristic other than characteristic two).
\end{prop}

It is worth noting that the Fano and non-Fano matroid instances are almost exactly the same, only differing in the role of set $N_{7}=\{4,5,6\}$. While set $N_{7}$ is a circuit set for the Fano matroid, it is a basis set for the non-Fano matroid.

\section{The Class $p$-Fano and $p$-non-Fano Matroid Instances}\label{sec:p-Fano-non-Fano}

\subsection{The Class $p$-Fano Matroid Instances $\mathcal{N}_{F}(p)$} \label{sub:p-Fano}


\begin{defn}[Class $p$-Fano Matroid Instances $\mathcal{N}_{F}(p)$ \cite{Bernt-matroid}]
Matroid instance $\mathcal{N}_{F}(p)$ is referred to as the class $p$-Fano matroid instance if (i) $n=2p+3$, $f(\mathcal{N}_{F}(p))=p+1$, (ii) $N_{0}=[p+1]\in \mathcal{B}$ and (iii) $N_{i}\in \mathcal{C}, \forall i\in [n]$, where
\begin{equation} \label{eq:Fano-circuit-N_i}
N_{i}=
  \left \{
    \begin{array}{ccc}
     &([p+1]\backslash \{i\})\cup \{n-i\},\ \ \ \ \ \ \ i\in [p+1], \ \ \ \ \ \ \ \ \ \ \ \ \ \ \ \ \ \ \ \ \ \ \
     \\
     \\
     &\{i, n-i, n\},\ \ \ \ \ \ \ \ \ \ \ \ \ \ \ \ \ \ \ \ i\in [p+2:2p+2], \ \ \ \ \ \ \ \ \ \ \ \ \ \ \
     \\
     \\
     &[p+2:2p+2], \ \ \ \ \ \ \ \ \ \ \ \ \ \ \ \ i=n=2p+3. \ \ \ \ \ \ \ \ \ \ \ \ \ \ \ \ \
    \end{array}
  \right.
\end{equation}
\end{defn}

\begin{exmp}[$\mathcal{N}_{F}(2)$]
It can be verified that $\mathcal{N}_{F}(2)=\mathcal{N}_{F}$. In other words, the Fano matroid instance is a special case of the class $p$-Fano matroid for $p=2$. 
\end{exmp}

\begin{exmp}[$\mathcal{N}_{F}(3)$ \cite{Arman-4-3}]
Consider the class $p$-Fano matroid for $p=3$. It can be seen that for matroid $\mathcal{N}_{F}(3)$, $n=9$, $f(\mathcal{N}_{F}(3))=4$, set $N_{0}=[4]$ is a basis set, and the following sets $N_{i}, i\in [9]$ are circuit sets:
\begin{align}
    N_{1}&=\{1,2,3,5\},\ \ N_{2}=\{1,2,4,6\},\ \ N_{3}=\{1,3,4,7\},
    \nonumber
    \\
    N_{4}&=\{2,3,4,8\},\ \ N_{5}=\{1,8,9\},\ \ \ \ \ N_{6}=\{2,7,9\},
    \nonumber
    \\
    N_{7}&=\{3,6,9\},\ \ \ \ \ N_{8}=\{4,5,9\},\ \ \ \ \ N_{9}=\{5,6,7,8\}.
    \label{eq:-N1-matroid}
\end{align}
\end{exmp}

\begin{thm}[\hspace{-0.05 ex}\cite{Bernt-matroid, pena-matroid}] \label{thm:p-Fano-matroid}
The class $p$-Fano Matroid $\mathcal{N}_{F}(p)$ is linearly representable over $\mathbb{F}_{q}$ iff field $\mathbb{F}_{q}$ has characteristic $p$. 
\end{thm}

The proof can be concluded from Propositions \ref{prop:Fano-necessary} and \ref{prop:Fano-sufficient}, which establish the necessary and sufficient conditions, respectively.
We provide a completely independent alternative proof of Proposition \ref{prop:Fano-necessary}, which exclusively relies on matrix manipulations instead of the more involved concepts of number theory and matroid theory used in \cite{Bernt-matroid} and \cite{pena-matroid}.

\begin{prop} \label{prop:Fano-necessary}
The class $p$-Fano matroid $\mathcal{N}_{F}(p)$ is linearly representable over $\mathbb{F}_{q}$ only if field $\mathbb{F}_{q}$ does have characteristic $p$.
\end{prop}

\begin{proof}
Since each $N_{i}, i\in [n]$ in \eqref{eq:Fano-circuit-N_i} is a circuit set, according to \eqref{eq:circuit-set-M-invertible}, we have
\begin{align}
    &\boldsymbol{H}^{\{n-i\}}=\sum_{j\in [p+1]\backslash \{i\}} \boldsymbol{H}^{\{j\}} \boldsymbol{M}_{n-i,j}, \ \ \ \ i\in [p+1],
    \label{eq:H-with-H-1}
    \\
    \nonumber
    \\
    &\boldsymbol{H}^{\{n\}}=\boldsymbol{H}^{\{i\}}\boldsymbol{M}_{n,i}+ \boldsymbol{H}^{\{n-i\}}\boldsymbol{M}_{n, n-i}, \ \ \ \ i\in [p+1],
    \label{eq:H-with-H-2}
    \\
    \nonumber
    \\
    &\boldsymbol{H}^{\{2p+2\}}=\sum_{i\in [p+1]\backslash \{1\}} \boldsymbol{H}^{\{n-i\}}\boldsymbol{M}_{2p+2,n-i},
    \label{eq:H-with-H-3}
\end{align}
where all the matrices $\boldsymbol{M}_{n-i,j}, j\in [p+1]\backslash \{i\}, i\in [p+1]$, $\boldsymbol{M}_{n,i}, \boldsymbol{M}_{n, n-i}, i\in [p+1]$ and $\boldsymbol{M}_{2p+2,n-i}, i\in [p+1]\backslash \{1\}$ are invertible.
Now, we replace $\boldsymbol{H}^{\{n-i\}}$ in \eqref{eq:H-with-H-2} with its equal term in \eqref{eq:H-with-H-1}, leading to
\begin{align}
    \boldsymbol{H}^{\{n\}}&=\boldsymbol{H}^{\{i\}}\boldsymbol{M}_{n,i} + \Big ( \sum_{j\in [p+1]\backslash \{i\}} \boldsymbol{H}^{\{j\}} \boldsymbol{M}_{n-i,j} \Big )\  \boldsymbol{M}_{n, n-i}
    \nonumber
    \\
    \nonumber
    \\
    &=[\boldsymbol{H}^{\{1\}},\dots,\boldsymbol{H}^{\{i-1\}},\boldsymbol{H}^{\{i\}},\boldsymbol{H}^{\{i+1\}},\dots, \boldsymbol{H}^{\{p+1\}}]
    \nonumber
    \\
    \nonumber
    \\
    & \ \ \ \ \ \ \ \ \ \ \ \ \ \ \times
    \underbrace{
    \begin{bmatrix}
    \boldsymbol{M}_{n-i,1}\ \boldsymbol{M}_{n,n-i}
    \\
    \vdots
    \\
    \boldsymbol{M}_{n-i,i-1}\ \boldsymbol{M}_{n,n-i}
    \\
    \\
    \boldsymbol{M}_{n,i}
    \\
    \\
    \boldsymbol{M}_{n-i,i+1}\ \boldsymbol{M}_{n,n-i}
    \\
    \vdots
    \\
    \boldsymbol{M}_{n-i,p+1}\ \boldsymbol{M}_{n,n-i}
    \end{bmatrix}
    }_\textrm{$\boldsymbol{M}_{i}^{\ast}$}
    \nonumber
    \\
    \nonumber
    \\
    &=\boldsymbol{H}^{[p+1]}\boldsymbol{M}_{i}^{\ast}.
    \nonumber
\end{align}
Thus, we have
\begin{align}
    \boldsymbol{H}^{\{n\}}=\boldsymbol{H}^{[p+1]}\boldsymbol{M}_{i}^{\ast}, \ \ \ \forall i\in [p+1].
\end{align}
Since $\boldsymbol{H}^{[p+1]}$ is full-rank, we must have
\begin{equation} \label{eq:Mstar-Mstar}
    \boldsymbol{M}_{1}^{\ast}=\boldsymbol{M}_{2}^{\ast}=\dots =\boldsymbol{M}_{p+1}^{\ast}.
\end{equation}
Consequently, for any $i\in [p+1]$, \eqref{eq:Mstar-Mstar} will lead to
\begin{align}
    &\boldsymbol{M}_{n,i}=\boldsymbol{M}_{n-j,i}\ \boldsymbol{M}_{n,n-j}, \ \ \ \forall i,j\in [p+1], \ i\neq j, 
    \label{eq:Msq=i-j-1}
    \\
    \nonumber
    \\
    \Rightarrow&
    \boldsymbol{M}_{n,i^{\prime}}=\boldsymbol{M}_{n-j,i^{\prime}}\ \boldsymbol{M}_{n,n-j}=\boldsymbol{M}_{n-l,i^{\prime}}\ \boldsymbol{M}_{n,n-l}, 
    \nonumber
    \\
    &\ \ \ \ \ \ \ \forall i^{\prime},j,l\in [p+1], \ i^{\prime}, j, l \ \text{ are distinct},  
    \label{eq:Msq=i-j-2}
    \\
    \nonumber
    \\
    \Rightarrow&
    \boldsymbol{M}_{n,n-l}=\boldsymbol{M}_{n-l,i^{\prime}}^{-1}\ \boldsymbol{M}_{n-j,i^{\prime}}\ \boldsymbol{M}_{n,n-j}, 
    \nonumber
    \\
    &\ \ \ \ \ \ \ \forall i^{\prime},j,l\in [p+1], \ i^{\prime}, j, l \ \text{are distinct},  
    \label{eq:Msq=i-j-3}
    \\
    \nonumber
    \\
    \Rightarrow&
    \boldsymbol{M}_{n-l,i_{1}}^{-1}\ \boldsymbol{M}_{n-j,i_{1}}\boldsymbol{M}_{n,n-j}=\boldsymbol{M}_{n-l,i_{2}}^{-1}\ \boldsymbol{M}_{n-j,i_{2}}\boldsymbol{M}_{n,n-j}, 
    \nonumber
    \\
    &\ \ \ \ \ \ \forall i_{1}, i_{2}, j, l\in [p+1], \ i_{1}, i_{2}, j, l \ \text{ are distinct},
    \label{eq:condition-1-i,j,l}
    \\
    \nonumber
    \\
    \Rightarrow&
    \boldsymbol{M}_{n-l,i_{1}}^{-1}\ \boldsymbol{M}_{n-j,i_{1}}=\boldsymbol{M}_{n-l,i_{2}}^{-1}\ \boldsymbol{M}_{n-j,i_{2}}, 
    \nonumber
    \\
    & \ \ \ \ \ \forall i_{1}, i_{2}, j, l\in [p+1], \ i_{1}, i_{2}, j, l \ \text{ are distinct}. 
    \label{eq:condition-1-i,l}
\end{align}
On the other hand, by replacing $\boldsymbol{H}^{\{n-i\}}$ in \eqref{eq:H-with-H-3} with its equal term in \eqref{eq:H-with-H-1}, we have
\begin{align}
    \boldsymbol{H}^{\{2p+2\}}&= \sum_{i\in [p+1]\backslash \{1\}} \boldsymbol{H}^{\{n-i\}} \boldsymbol{M}_{2p+2,n-i}
    \nonumber
    \\
    \nonumber
    \\
    &=\sum_{i\in [p+1]\backslash \{1\}} \Big (\sum_{j\in [p+1]\backslash \{i\}} \boldsymbol{H}^{\{j\}} \boldsymbol{M}_{n-i,j} \Big )\boldsymbol{M}_{2p+2,n-i}
    \nonumber
    \\
    \nonumber
    \\
    &= \boldsymbol{H}^{\{1\}}\ \sum_{i\in [p+1]\backslash \{1\}} \boldsymbol{M}_{n-i,1} \boldsymbol{M}_{2p+2,n-i} \ \ +
    \nonumber
    \\
    &\ \ \sum_{j\in [p+1]\backslash \{1\}} \boldsymbol{H}^{\{j\}} \sum_{i\in [p+1]\backslash \{1,j\}} \boldsymbol{M}_{n-i,j}\ \boldsymbol{M}_{2p+2,n-i}.
    \label{eq-H^[2p+2]-2}
\end{align}
Moreover, in \eqref{eq:H-with-H-1} for $i=1$, we have (note $n-1=2p+2$)
\begin{align}
    \boldsymbol{H}^{\{2p+2\}}&=\sum_{j\in [p+1]\backslash \{1\}} \boldsymbol{H}^{\{j\}} \boldsymbol{M}_{2p+2,j}
    \nonumber
    \\
    \nonumber
    \\
    &=\boldsymbol{H}^{\{1\}} \times \boldsymbol{0}
    +\sum_{j\in [p+1]\backslash \{1\}} \boldsymbol{H}^{\{j\}} \boldsymbol{M}_{2p+2,j}.
    \label{eq-H^[2p+2]-1}
\end{align}
Now, since both \eqref{eq-H^[2p+2]-1} and \eqref{eq-H^[2p+2]-2} are equal to $\boldsymbol{H}^{\{2p+2\}}$, and each $\boldsymbol{H}^{\{j\}}, \in [p+1]$ is linearly independent, the coefficients of $\boldsymbol{H}^{\{j\}}, \in [p+1]$ must be equal. Thus, 
\begin{align}
    &\sum_{i\in [p+1]\backslash \{1\}} \boldsymbol{M}_{n-i,1} \boldsymbol{M}_{2p+2,n-i}=0,
    \label{eq:last-1}
    \\
    \nonumber
    \\
    &\sum_{i\in [p+1]\backslash \{1,j\}} \boldsymbol{M}_{n-i,j}\ \boldsymbol{M}_{2p+2,n-i}=\boldsymbol{M}_{2p+2,j},
    \nonumber
    \\
    &
    \ \ \ \ \ \ \ \ \ \ \ \ \ \ \  \forall j\in [p+1]\backslash \{i\}.
    \label{eq:11}
\end{align}
From \eqref{eq:11}, for $j_{1}, j_{2}\in [p+1]\backslash \{i\}, j_{1}\neq j_{2}$, we have
\begin{align}
    \sum_{i\in [p+1]\backslash \{1,j_{1}\}} \boldsymbol{M}_{n-i,j_{1}}\ \boldsymbol{M}_{2p+2,n-i}&=\boldsymbol{M}_{2p+2,j_{1}},
    \label{eq:pr:01}
    \\
    \nonumber
    \\
    \sum_{i\in [p+1]\backslash \{1,j_{2}\}} \boldsymbol{M}_{n-i,j_{2}}\ \boldsymbol{M}_{2p+2,n-i}&=\boldsymbol{M}_{2p+2,j_{2}}.
     \label{eq:pr:02}
\end{align}
Now, for $l\in [p+1]\backslash \{1,j_{1},j_{2}\}$, multiplying \eqref{eq:pr:01} and \eqref{eq:pr:02}, respectively, by $ \boldsymbol{M}_{n-l,j_{1}}^{-1}$ and $\boldsymbol{M}_{n-l,j_{2}}^{-1}$ will lead to
\begin{align}
    &\sum_{i\in [p+1]\backslash \{1,j_{1}\}} \boldsymbol{M}_{n-l,j_{1}}^{-1}\  \boldsymbol{M}_{n-i,j_{1}}\ \boldsymbol{M}_{2p+2,n-i}=
    \nonumber
    \\
    &\boldsymbol{M}_{n-l,j_{1}}^{-1} \ \boldsymbol{M}_{2p+2,j_{1}},
    \label{eq:proof:03}
    \\
    \nonumber
    \\
    &\sum_{i\in [p+1]\backslash \{1,j_{2}\}} \boldsymbol{M}_{n-l,j_{2}}^{-1}\ \boldsymbol{M}_{n-i,j_{2}}\ \boldsymbol{M}_{2p+2,n-i}=
    \nonumber
    \\
    &\boldsymbol{M}_{n-l,j_{2}}^{-1} \ \boldsymbol{M}_{2p+2,j_{2}}.
    \label{eq:proof:04}
\end{align}
Now, according to \eqref{eq:condition-1-i,l}, for $  i_{1}=j_{1}, j=1$ and $i_{2}=j_{2}$, we have  $\boldsymbol{M}_{n-l,j_{1}}^{-1} \ \boldsymbol{M}_{2p+2,j_{1}}=\boldsymbol{M}_{n-l,j_{2}}^{-1} \ \boldsymbol{M}_{2p+2,j_{2}}$. Thus, the left hand side (LHS) of \eqref{eq:proof:03} and \eqref{eq:proof:04} are equal. Hence, after removing common terms on the LHS of \eqref{eq:proof:03} and \eqref{eq:proof:04}, we get
\begin{align}
    &\boldsymbol{M}_{n-l,j_{2}}^{-1}\ \boldsymbol{M}_{n-j_{1},j_{2}}\ \boldsymbol{M}_{2p+2,n-j_{1}}=
    \nonumber
    \\
    &\boldsymbol{M}_{n-l,j_{1}}^{-1}\  \boldsymbol{M}_{n-j_{2},j_{1}}\ \boldsymbol{M}_{2p+2,n-j_{2}}.
    \label{eq:12}
\end{align}
Now, based on \eqref{eq:condition-1-i,l}, for $i_1=j_2, j = j_1, i_2=1$, we have $\boldsymbol{M}_{n-l,j_{2}}^{-1}\ \boldsymbol{M}_{n-j_{1},j_{2}}=\boldsymbol{M}_{n-l,1}^{-1}\ \boldsymbol{M}_{n-j_{1},1}$ and for $i_1=j_1, j=j_2, i_2=1$, we have $\boldsymbol{M}_{n-l,j_{1}}^{-1}\  \boldsymbol{M}_{n-j_{2},j_{1}}=\boldsymbol{M}_{n-l,1}^{-1}\  \boldsymbol{M}_{n-j_{2},1}$. Thus, \eqref{eq:12} will lead to
\begin{align}
    &\boldsymbol{M}_{n-l,1}^{-1}\ \boldsymbol{M}_{n-j_{1},1}\ \boldsymbol{M}_{2p+2,n-j_{1}}=
    \nonumber
    \\
    &\boldsymbol{M}_{n-l,1}^{-1}\  \boldsymbol{M}_{n-j_{2},1}\ \boldsymbol{M}_{2p+2,n-j_{2}}.
    \label{eq:condition-1-leadto-1and1}
\end{align}
Therefore,
\begin{align}
    \ \boldsymbol{M}_{n-j_{1},1}\ \boldsymbol{M}_{2p+2,n-j_{1}}&=  \boldsymbol{M}_{n-j_{2},1}\ \boldsymbol{M}_{2p+2,n-j_{2}},
    \label{eq:123}
\end{align}
which means that all the terms $\boldsymbol{M}_{n-i,1}\boldsymbol{M}_{2p+2,n-i}, i\in [p+1]\backslash \{1\}$ in \eqref{eq:last-1} are equal. Let $\boldsymbol{M}_{n-i,1}\boldsymbol{M}_{2p+2,n-i}=\boldsymbol{T}, i\in [p+1]\backslash \{1\}$. Then,
\begin{align}
    \sum_{i\in [p+1]\backslash \{1\}} \boldsymbol{M}_{n-i,1}\boldsymbol{M}_{2p+2,n-i}
    =
    \Big (\sum_{i\in [p+1]\backslash \{1\}} \boldsymbol{I}_{t}\Big )\ \boldsymbol{T}=0.
    \label{eq:one-to-last-necessary-p-Fano}
\end{align}
Now, since matrix $\boldsymbol{T}$ is invertible, \eqref{eq:one-to-last-necessary-p-Fano} requires that
\begin{equation}
    \sum_{i\in [p]} \boldsymbol{I}_{t} = 0,
\end{equation}
which is possible only over fields with characteristic $p$. This completes the proof.
\end{proof}

In the following, Lemma \ref{lem:matrix-F-(p+1)-p} is provided for proving Propositions \ref{prop:Fano-sufficient} and \ref{prop:non-Fano-sufficient}.

\begin{lem} \label{lem:matrix-F-(p+1)-p}
For the following matrix, if $\mathbb{F}_{q}$ has characteristic $p$, then $[p+1]$ is a circuit set. Otherwise, $[p+1]$ is a basis set.
\begin{align}
    \begin{bmatrix}
    1       & 1      & 1      & \hdots & 1      & 1      & 0      \\
    1       & 1      & 1      & \hdots & 1      & 0      & 1      \\
    1       & 1      & 1      & \hdots & 0      & 1      & 1      \\
    \vdots  & \vdots & \vdots & \ddots & \vdots & \vdots & \vdots \\
    1       & 1      & 0      & \hdots & 1      & 1      & 1      \\
    1       & 0      & 1      & \hdots & 1      & 1      & 1      \\
    0       & 1      & 1      & \hdots & 1      & 1      & 1      
    \end{bmatrix}\in \mathbb{F}_{q}^{(p+1)\times (p+1)}.
\end{align}
\end{lem}

\begin{proof}
It can easily be observed that running the Gaussian elimination technique on the first column shows that the first $p$ columns are linearly independent, as it results in
\begin{align}
    \begin{bmatrix}
    \textcolor{red}{{\large \textcircled{\small 1}}}       & 1                                                      & 1                                                      & \hdots & 1      & 1      & 0      
    \\
    0                                                      & 0                                                      & 0                                                      & \hdots & 0      & \textcolor{red}{{\large \textcircled{\small -1}}}     & 1      
    \\
    0                                                      & 0                                                      & 0                                                      & \hdots & \textcolor{red}{{\large \textcircled{\small -1}}}     & 0      & 1      
    \\
    \vdots                                                 & \vdots                                                 & \vdots & \ddots & \vdots & \vdots & \vdots 
    \\
    0                                                      & 0                                                      & \textcolor{red}{{\large \textcircled{\small -1}}}     & \hdots & 0      & 0      & 1      
    \\
    0                                                      & \textcolor{red}{{\large \textcircled{\small -1}}}      & 0                                                      & \hdots & 0      & 0      & 1      
    \\
    0                                                      & 1                                                      & 1                                                      & \hdots & 1      & 1      & 1      
    \end{bmatrix}.
\end{align}
Running the Gaussian elimination on the last row will lead to
\begin{align}
    \begin{bmatrix}
    \textcolor{red}{{\large \textcircled{\small 1}}}       & 1                                                      & 1                                                      & \hdots & 1      & 1      & 0      
    \\
    0                                                      & 0                                                      & 0                                                      & \hdots & 0      & \textcolor{red}{{\large \textcircled{\small -1}}}     & 1      
    \\
    0                                                      & 0                                                      & 0                                                      & \hdots & \textcolor{red}{{\large \textcircled{\small -1}}}     & 0      & 1      
    \\
    \vdots                                                 & \vdots                                                 & \vdots & \ddots & \vdots & \vdots & \vdots 
    \\
    0                                                      & 0                                                      & \textcolor{red}{{\large \textcircled{\small -1}}}     & \hdots & 0      & 0      & 1      
    \\
    0                                                      & \textcolor{red}{{\large \textcircled{\small -1}}}      & 0                                                      & \hdots & 0      & 0      & 1      
    \\
    0                                                      & 0                                                      & 0                                                      & \hdots & 0      & 0      &   \textcolor{red}{1+(p-1)}    
    \end{bmatrix}.
\end{align}
Now, it can be seen that if the characteristic of field $\mathbb{F}_{q}$ is $p$, then $1+(p-1)=p=0$, which means that the last column is linearly dependent on the first $p$ columns, and thus, set $[p+1]$ is a circuit set. However, if the characteristic of $\mathbb{F}_{q}$ is not $p$, then $1+(p-1)=p\neq 0$, which means that the last column is also a pivot column, and thus, $[p+1]$ is a basis set.
\end{proof}

\begin{prop} \label{prop:Fano-sufficient}
The class $p$-Fano matroid $\mathcal{N}_{F}(p)$ has scalar ($t=1$) linear representation over fields with characteristic $p$.
\end{prop}

\begin{proof}
We show that matrix $\boldsymbol{H}_{p}\in \mathbb{F}_{q}^{(p+1)\times n}$, shown in Figure \ref{fig:H-class-fano-nonfano}, is a scalar linear representation of $\mathcal{N}_{F}(p)$ if field $\mathbb{F}_{q}$ does have characteristic $p$.
\begin{figure*}
\centering
\subfloat
{
$
\label{eq:p-Fano-N-p-H}
\boldsymbol{H}_{p}^{\{i\}}=
  \left \{
    \begin{array}{ccc}
     &\boldsymbol{I}_{p+1}^{\{i\}},\ \ \ \ \ \ \ \ \ \ \ \ \ \ \ \ \ \ \ \ \ \ \ \ \ \ \ \ \ i\in [p+1], \ \ \ \ \ \ \ \ \ \ \ 
     \\
     \\
     &\sum_{j\in [p+1]\backslash \{n-i\}} \boldsymbol{I}_{p+1}^{\{j\}}, \ \ \ \ \ \ \ \ \ \ i\in [p+2:2p+2],
     \\
     \\
     &\sum_{j\in [p+1]} \boldsymbol{I}_{p+1}^{\{j\}}, \ \ \ \ \ \ \ \ \ \ \ \ \ \ \ \ \ \ \ i=n=2p+3. \ \ \ \
    \end{array}
  \right.
$
}
\\
\vspace{2ex}
\subfloat
{
$\boldsymbol{H}_{p}=
\begin{blockarray}{cccccccccccc}
              & 
              \textcolor{blue}{{\tiny 1}}        & 
              \textcolor{blue}{{\tiny 2}}        & 
              \textcolor{blue}{{\tiny \hdots}}   & 
              \textcolor{blue}{{\tiny p}}        &
              \textcolor{blue}{{\tiny p+1}}      &
              \textcolor{blue}{{\tiny p+2}}      & 
              \textcolor{blue}{{\tiny p+3}}      & 
              \textcolor{blue}{{\tiny \hdots}}   & 
              \textcolor{blue}{{\tiny 2p+1}}     &
              \textcolor{blue}{{\tiny 2p+2}}     &
              \textcolor{blue}{{\tiny 2p+3}}  \\ 
\begin{block}{c[ccccc|ccccc|c]}
              \textcolor{blue}{{\tiny 1}} & 
              1 & 0 & \hdots & 0 & 0 &
              1 & 1 & \hdots & 1 & 0 &
              1
              \\
              \textcolor{blue}{{\tiny 2}} & 
              0 & 1 & \hdots & 0 & 0 &
              1 & 1 & \hdots & 0 & 1 &
              1
              \\
              \textcolor{blue}{{\tiny \vdots}} & 
              \vdots & \vdots & \ddots & \vdots & \vdots & 
              \vdots & \vdots & \ddots & \vdots & \vdots &
              \vdots &
              \\
              \textcolor{blue}{{\tiny p}} & 
              0 & 0 & \hdots & 1 & 0 &
              1 & 0 & \hdots & 1 & 1 &
              1 &
              \\
              \textcolor{blue}{{\tiny p+1}} & 
              0 & 0 & \hdots & 0 & 1 & 
              0 & 1 & \hdots & 1 & 1 & 
              1 
              \\
             \end{block}
\end{blockarray}
$
}
    \caption{$\boldsymbol{H}_{p}\in \mathbb{F}_{q}^{(p+1)\times n}$: If $\mathbb{F}_{q}$ has characteristic $p$, then $\boldsymbol{H}_{p}$ is a scalar linear representation of the class $p$-Fano matroid $\mathcal{N}_{F}(p)$, and if $\mathbb{F}_{q}$ has any characteristic other than characteristic $p$, then $\boldsymbol{H}_{p}$  is a scalar linear representation of the class $p$-non-Fano matroid $\mathcal{N}_{nF}(p)$.}
    \label{fig:H-class-fano-nonfano}
\end{figure*}

\begin{itemize}[leftmargin=*]
    \item It can be seen that set $N_{0}=[p+1]$ is a basis set of $\boldsymbol{H}_{p}$ as
\begin{equation}
    \mathrm{rank}(\boldsymbol{H}_{p}^{[p+1]})=\mathrm{rank}(\boldsymbol{I}_{p+1})=p+1.
\end{equation}
    \item Now, we show that each set $N_{i}, i\in n$ in \eqref{eq:Fano-circuit-N_i} is a circuit set of $\boldsymbol{H}_{p}$. 
    \\
    First, we begin with sets $N_{i}=([p+1]\backslash\{i\})\cup \{n-i\}, i\in [p+2:2p+2]$. Since set $[p+1]$ is a basis set of $\boldsymbol{H}_{p}$, set $[p+1]\backslash \{i\}$ forms an independent set of $\boldsymbol{H}_{p}$. Moreover, we have 
    \begin{equation}
    \boldsymbol{H}_{p}^{\{n-i\}}= \sum_{j\in ([p+1]\backslash\{i\})} \boldsymbol{I}_{p+1}^{\{j\}}=\sum_{j\in ([p+1]\backslash\{i\})} \boldsymbol{H}_{p}^{\{j\}}.
    \end{equation}
    Thus, based on \eqref{eq:circuit-set-M-invertible}, each set $N_{i}=([p+1]\backslash\{i\})\cup \{n-i\}, i\in [p+2:2p+2]$ is a circuit set.
    
    \item For sets $N_{i}=\{i, n-i, n\}, i\in [p+2:2p+2]$, from Figure \ref{fig:H-class-fano-nonfano}, it can be observed that set $\{i, n-i\}$ is an independent set of $\boldsymbol{H}_{p}$. Furthermore, we have
    \begin{align}
    \boldsymbol{H}_{p}^{\{n\}}
    &=\sum_{j\in [p+1]} \boldsymbol{I}_{p+1}^{\{j\}}
    \nonumber
    \\
    &=\boldsymbol{I}_{p+1}^{\{i\}} + \sum_{j\in [p+1]\backslash \{i\}} \boldsymbol{I}_{p+1}^{\{j\}}
    \nonumber
    \\
    &=\boldsymbol{H}_{p}^{\{i\}}+\boldsymbol{H}_{p}^{\{n-i\}}.
    \end{align}
    Thus, based on equation \eqref{eq:circuit-set-M-invertible}, each set $\{i, n-i, n\}$ is a circuit set of $\boldsymbol{H}_{p}$.
    
    \item Finally, based on Lemma \ref{lem:matrix-F-(p+1)-p}, set $[p+2:2p+2]$ is a circuit set of $\boldsymbol{H}_{p}$ over the fields with characteristic $p$. This completes the proof.
\end{itemize}
\end{proof}

\subsection{The Class $p$-non-Fano Matroid Instances $\mathcal{N}_{nF}(p)$} \label{sub:p-non-Fano}

\begin{defn}[Class $p$-non-Fano Matroid Instances $\mathcal{N}_{nF}(p)$  \cite{Bernt-matroid}]
Matroid $\mathcal{N}_{nF}(p)$ is referred to as the class $p$-non-Fano instance if (i) $n=2p+3$, $f(\mathcal{N}_{nF}(p))=p+1$, (ii) $N_{0}=[p+1]\in \mathcal{B}$, $N_{n}=[p+2:2p+2]\in \mathcal{B}$ and (iii) $N_{i}\in \mathcal{C}, \forall i\in [n-1]$ where
\begin{equation} \label{eq:nonFano-circuit-N_i}
N_{i}=
  \left \{
    \begin{array}{cc}
     ([p+1]\backslash \{i\})\cup \{n-i\}, \ i\in [p+1], \ \ \ \ \ \ \ \ \ \ \ &
     \\
     \\
     \{i, n-i, n\},\ \ \ \ \ \ \ \ \ \ \ \ \ \ i\in [p+2:2p+2].&
    \end{array}
  \right.
\end{equation}
\end{defn}

It is worth noting that the class $p$-Fano and the class $p$-non-Fano matroid instances are almost exactly the same, only differing in the role of set $N_{n}=[p+2:2p+2]$. While set $N_{n}$ is a circuit set for the class $p$-Fano matroid, it is a basis set for the class $p$-non-Fano matroid.

\begin{exmp}[$\mathcal{N}_{nF}(2)$]
It can be verified that $\mathcal{N}_{nF}(2)=\mathcal{N}_{nF}$. In other words, the non-Fano matroid instance is a special case of the class $p$-non-Fano matroid for $p=2$. 
\end{exmp}

\begin{exmp}[$\mathcal{N}_{nF}(3)$ \cite{Arman-4-3}]
Consider the class $p$-non-Fano matroid for $p=3$. It can be seen that for matroid $\mathcal{N}_{nF}(3)$, $n=9$, $f(\mathcal{N}_{nF}(3))=4$, sets $N_{0}=[4]$ and $N_{9}=\{5,6,7,8\}$ are basis sets, and sets $N_{i}, i\in [8]$ in \eqref{eq:-N1-matroid} are all circuit sets.
\end{exmp}

\begin{thm}[\hspace{-0.05 ex}\cite{Bernt-matroid, pena-matroid}] \label{thm:p-non-Fano-matroid}
The class $p$-non-Fano matroid $\mathcal{N}_{nF}(p)$ is linearly representable over $\mathbb{F}_{q}$ if and only if field $\mathbb{F}_{q}$ does have any characteristic other than characteristic $p$. 
\end{thm}

The proof can be concluded from Propositions \ref{prop:non-Fano-necessary} and \ref{prop:non-Fano-sufficient},  which establish the necessary and sufficient conditions, respectively.
We provide a completely independent alternative proof of Proposition \ref{prop:non-Fano-necessary}, which exclusively relies on matrix manipulations instead of the more involved concepts of number theory and matroid theory used in \cite{Bernt-matroid} and \cite{pena-matroid}.

\begin{prop} \label{prop:non-Fano-necessary}
The class $p$-non-Fano matroid $\mathcal{N}_{nF}(p)$ is linearly representable over $\mathbb{F}_{q}$ only if field $\mathbb{F}_{q}$ does have any characteristic other than characteristic $p$.
\end{prop}
\begin{proof}
Since all sets $N_{i}, i\in \{0\}\cup [n-1]$ in $\mathcal{N}_{nF}(p)$ are exactly the same as the sets $N_{i}, i\in \{0\}\cup [n-1]$ in $\mathcal{N}_{F}(p)$, \eqref{eq:Msq=i-j-3} will also hold for $\mathcal{N}_{nF}(p)$.
\\
Now, since each $\boldsymbol{M}_{n-j,i}$ is invertible, from \eqref{eq:Msq=i-j-3}, the column space of all $\boldsymbol{M}_{n,i}$'s for $i\in [p+1]$ will be equal. Thus, we have
\begin{equation}
    \mathrm{rank}\ \Big [\boldsymbol{M}_{n,1}\ | \ \boldsymbol{M}_{n,2}\ | \ \hdots \ | \  \boldsymbol{M}_{n,p+1} \Big ]= \mathrm{rank} \ \Big [\boldsymbol{M}_{n, i} \Big ], 
\end{equation}
which requires each matrix $\boldsymbol{M}_{n,i}, i\in [p+1]$ to be invertible, since otherwise it causes $\boldsymbol{H}^{\{n\}}$ to be non-invertible (which leads to a contradiction based on Remark \ref{rem:1}).
\\
Now, according to \eqref{eq:Msq=i-j-2}, we have
\begin{align}
    &(p-1)\boldsymbol{M}_{n,i}=(p-1)\boldsymbol{M}_{n-j,i}\ \boldsymbol{M}_{n,n-j}=
    \nonumber
    \\
    &\sum_{l\in [p+1]\backslash \{i,j\}} \boldsymbol{M}_{n-l,i}\ \boldsymbol{M}_{n,n-l}, \ \ \forall i,j\in [p+1], \ i\neq j.
    \label{eq:nonFano-nec-last-1-1}
\end{align}
Moreover, based on \eqref{eq:Msq=i-j-2}, since the terms $ \boldsymbol{M}_{n-l,i}\ \boldsymbol{M}_{n,n-l}, l\in [p+1]\backslash \{i\}$ are all equal, then assuming that the field has characteristic $p$ will result in
\begin{align}
    \boldsymbol{0}_{t}=\sum_{l\in [p+1]\backslash \{i\}} \boldsymbol{M}_{n-l,i}\ \boldsymbol{M}_{n,n-l}.
    \label{eq:nonFano-nec-last-1-2}
\end{align}

Now, combining \eqref{eq:nonFano-nec-last-1-1} and \eqref{eq:nonFano-nec-last-1-2} will lead to

\begin{align}
    &(p-1)
    \begin{bmatrix}
    \boldsymbol{M}_{n-(p+1),1}
    \\
    \boldsymbol{M}_{n-(p+1),2}
    \\
    \vdots
    \\
    \boldsymbol{M}_{n-(p+1),p}
    \\
    \boldsymbol{0}_{t}
    \end{bmatrix}
    \boldsymbol{M}_{n, n-(p+1)}
    =
    \nonumber
    \\
    &\begin{bmatrix}
    \boldsymbol{0}_{t}
    \\
    \boldsymbol{M}_{n-1,2}
    \\
    \vdots
    \\
    \boldsymbol{M}_{n-1,p}
    \\
    \boldsymbol{M}_{n-1,(p+1)}
    \end{bmatrix}
    \boldsymbol{M}_{n, n-1}
    +
    \dots
    +
    \begin{bmatrix}
    \boldsymbol{M}_{n-p,1}
    \\
    \boldsymbol{M}_{n-1p,2}
    \\
    \vdots
    \\
    \boldsymbol{0}_{t}
    \\
    \boldsymbol{M}_{n-p,(p+1)}
    \end{bmatrix}
    \boldsymbol{M}_{n, n-p}.
    \label{eq:proof:last}
\end{align}
Then, from left, we multiply all the terms in \eqref{eq:proof:last} by the vector $[\boldsymbol{H}^{\{1\}},\dots, \boldsymbol{H}^{\{n\}}]$, which according to \eqref{eq:H-with-H-1}, it will lead to
\begin{align}
    &(p-1)\boldsymbol{H}^{\{n-(p+1)\}}\boldsymbol{M}_{n, n-(p+1)}=
    \nonumber
    \\
    &\boldsymbol{H}^{\{n-1\}}\boldsymbol{M}_{n,n-1}+\dots + \boldsymbol{H}^{\{n-p\}}\boldsymbol{M}_{n, n-p}.
    \label{eq:nonFano-nec-last-4-1}
\end{align}
Since each $\boldsymbol{M}_{n, n-i}, i\in [p+1]$ is invertible, based on \eqref{eq:circuit-set-M-invertible} and \eqref{eq:nonFano-nec-last-4-1}, it implies that set $[p+2:2p+2]$ forms a circuit set, which contradicts the fact that set $N_{n}=[p+2:2p+2]$ is a basis set of $\mathcal{N}_{nF}(p)$. This completes the proof.
\end{proof}

\begin{prop} \label{prop:non-Fano-sufficient}
The class $p$-non-Fano matroid $\mathcal{N}_{nF}(p)$ has scalar ($t=1$) linear representation over fields with any characteristic other than characteristic $p$.
\end{prop}

\begin{proof}
With the same arguments in the proof of Proposition \ref{prop:Fano-sufficient}, it can be shown that for matrix $\boldsymbol{H}_{p}\in \mathbb{F}_{q}^{(p+1)\times n}$, shown in Figure \ref{fig:H-class-fano-nonfano}, set $N_{0}$ is a basis set, and each set $N_{i}, i\in [n-1]$ is a circuit set of $\boldsymbol{H}_{p}$. Moreover, based on Lemma \ref{lem:matrix-F-(p+1)-p}, set $N_{n}=[p+2:2p+2]$ is a basis set over fields with any characteristic other than $p$. This completes the proof.
\end{proof}

\section{The Class $p$-Fano and $p$-non-Fano Index Coding Instance} \label{sec: index-coding-both}

\subsection{On the Reduction Process from Index Coding to Matroid} \label{sub:index-coding-reduction-matroid}
In this subsection, Lemmas \ref{lem:MAIS1}-\ref{lem:col(9)} establish reduction techniques to map specific constraints on the column space of the encoder matrix of an index coding instance to the constraints on the column space of the matrix, which is a linear representation of a matroid instance.

In this subsection, we assume that $M\subseteq [m]$, $i,l\in M$, and $j\in [m]\backslash M$.

\begin{lem}[\cite{Arman-4-3}]  \label{lem:MAIS1}
Assume $M$ is an acyclic set of $\mathcal{I}$. Then, the condition in \eqref{eq:dec-cond} for all $i\in M$ requires $\mathrm{rank} (\boldsymbol{H}^M)=|M|t$, implying that $M$ must be an independent set of $\boldsymbol{H}$.
\end{lem}

\begin{lem}[\cite{Arman-4-3}]  \label{lem:min-cyc}
Let $M$ be a minimal cyclic set of $\mathcal{I}$. To have $\mathrm{rank}(\boldsymbol{H}^{M})=(|M|-1)t$, $M$ must be a circuit set of $\boldsymbol{H}$.
\end{lem}

\begin{lem}[\cite{Arman-4-3}]   \label{lem:MAIS2}
Assume $M$ is an independent set of $\mathcal{I}$, and $j\in B_i,\forall i\in M\backslash \{l\}$ for some $l\in M$. Then, if $\mathrm{col}(\boldsymbol{H}^{\{j\}})\subseteq \mathrm{col}(\boldsymbol{H}^{M})$, we must have $\mathrm{col}(\boldsymbol{H}^{\{j\}})=\mathrm{col}(\boldsymbol{H}^{\{l\}})$.
\end{lem}

\begin{lem}[\cite{Arman-4-3}]   \label{lem:independent-cycle}
Let $M\subseteq[m]$ and $j\in [m]\backslash M$. Assume that
\begin{enumerate}[label=(\roman*)]
  \item $M$ is an independent set of $\boldsymbol{H}$,
  \item $\mathrm{col}(\boldsymbol{H}^{\{j\}})\subseteq \mathrm{col}(\boldsymbol{H}^{M})$,
  \item $M$ forms a minimal cyclic set of $\mathcal{I}$,
  \item $j\in B_{i}, \forall i\in M$.
\end{enumerate}
Now, the condition in \eqref{eq:dec-cond} for all $i\in [m]$ requires set $\{j\}\cup M$ to be a circuit set of $\boldsymbol{H}$.
\end{lem}

Now, we provide Lemmas \ref{lem:col(9)} and \ref{lem:i-n-i-n}, which will both be used in the proof of Proposition \ref{prop:I-fano-necessary}. Lemma \ref{lem:col(9)} generalizes Lemma 5 in \cite{Arman-4-3}.

\begin{lem}  \label{lem:col(9)}
Assume for matrix $\boldsymbol{H}\in \mathbb{F}_{q}^{(p+1)t\times nt}$, 
\begin{enumerate}[label=(\roman*)]
  \item set $[p+1]$ is a basis set,
  \item each set $[p+1]\backslash \{i\} \cup \{n-i\}, i\in [p+1]$ is a circuit set,
  \item $\mathrm{col}(\boldsymbol{H}^{\{n\}})\subseteq \mathrm{col}(\boldsymbol{H}^{\{i,n-i\}}), \forall i\in [p+1].$
        
\end{enumerate}
Then, each set $\{i, n-i, n\}, i\in [p+1]$ will be a circuit set.
\end{lem}

\begin{proof}
    Let $i\in [p+1]$. Since $\mathrm{col}(\boldsymbol{H}^{\{n\}})\subseteq \mathrm{col}(\boldsymbol{H}^{\{i,j-i\}})$, we have 
    \begin{align}
        \boldsymbol{H}^{\{n\}}=\boldsymbol{H}^{\{i\}}\boldsymbol{M}_{n,i}+\boldsymbol{H}^{\{n-i\}}\boldsymbol{M}_{n,n-i}.
        \label{eq:lem:last-conversion-9-1}
    \end{align}
    Besides, since set $[p+1]\backslash \{i\} \cup \{n-i\}$ is a circuit set, we must have
    \begin{align}
        \boldsymbol{H}^{\{n-i\}}=\sum_{j\in [p+1]\backslash \{i\}} \boldsymbol{H}^{\{i\}}\boldsymbol{M}_{n-i,j},
        \label{eq:lem:last-conversion-9-2}
    \end{align}
    where each matrix $\boldsymbol{M}_{n-j,j}, j\in [p+1]\backslash \{i\}$ is invertible. Thus, based on \eqref{eq:lem:last-conversion-9-1} and \eqref{eq:lem:last-conversion-9-2}, we have
    \begin{align}
        \boldsymbol{H}^{\{n\}} &= \boldsymbol{H}^{\{i\}}\boldsymbol{M}_{n,i}+(\sum_{j\in [p+1]\backslash \{i\}} \boldsymbol{H}^{\{j\}}\boldsymbol{M}_{n-i,j})\boldsymbol{M}_{n,n-i}
        \nonumber
        \\
        &=\boldsymbol{H}^{\{i\}}\boldsymbol{M}_{n,i}+ \sum_{j\in [p+1]\backslash \{i\}} \boldsymbol{H}^{\{j\}}\boldsymbol{M}_{n-i,j}^{\prime},
        \label{eq:lem3:1}
    \end{align}
    where $\boldsymbol{M}_{n-i,j}^{\prime} = \boldsymbol{M}_{n-i,j} \boldsymbol{M}_{n,n-i}, j\in [p+1]\backslash \{i\}$. 
    Now, let $l\in [p+1]\backslash \{i\} $. Since set $[p+1]\backslash \{l\}$ is a circuit set, we have 
    \begin{align}
        \boldsymbol{H}^{\{n-l\}}=\sum_{j\in [p+1]\backslash \{l\}} \boldsymbol{H}^{\{j\}}\boldsymbol{M}_{n-l,j},
        \label{eq:lem3:2}
    \end{align}
    where each $\boldsymbol{M}_{n-l,j}$ is invertible. Now, for set $\{l,n-l,l\}$, we have
    \begin{align}
        \mathrm{rank}(\boldsymbol{H}&^{\{l,n-l,n\}})
        =\mathrm{rank}([\boldsymbol{H}^{\{l\}}|\boldsymbol{H}^{\{n-l\}}|\boldsymbol{H}^{\{n\}}])
        \nonumber
        \\
        &=\mathrm{rank}([\boldsymbol{H}^{\{l\}}|\boldsymbol{H}^{\{n-l\}}|\boldsymbol{H}^{\{n\}}-\boldsymbol{H}^{\{l\}}\boldsymbol{M}_{n-i,l}^{\prime}])
        \nonumber
        \\
        &=t +\mathrm{rank}([\boldsymbol{H}^{\{n-l\}}|\boldsymbol{H}^{\{n\}}-\boldsymbol{H}^{\{l\}}\boldsymbol{M}_{n-i,l}^{\prime}]).
        \label{eq:lem3:3}
    \end{align}
    Now,
    since $\mathrm{col}(\boldsymbol{H}^{\{n\}})$ must be a subspace of $\mathrm{col}(\boldsymbol{H}^{\{l,n-l,n\}})$, we must have $\mathrm{rank}(\boldsymbol{H}^{\{l,n-l,n\}})=2t$. Thus, in \eqref{eq:lem3:3}, $\boldsymbol{H}^{\{n\}}-\boldsymbol{H}^{\{l\}}\boldsymbol{M}_{n-i,l}^{\prime}$ must be linearly dependent on $\boldsymbol{H}^{\{n-l\}}$, which based on \eqref{eq:lem3:1} and \eqref{eq:lem3:2} requires each $\boldsymbol{M}_{n,i}, \boldsymbol{M}_{n-i,j}^{\prime}, j\in [p+1]\backslash \{i,l\}$ to be invertible. 
    Moreover, each $\boldsymbol{M}_{n-i,j}^{\prime}, j\in [p+1]\backslash \{i,l\}$ is invertible only if $\boldsymbol{M}_{n,n-i}$ is invertible. Thus, since both $\boldsymbol{M}_{n,i}$ and $\boldsymbol{M}_{n,n-i}$ are invertible, set $\{i,n-i,n\}, i\in [p+1]$ forms a circuit set of $\boldsymbol{H}$.
\end{proof}

\begin{lem} \label{lem:i-n-i-n}
    Let $M_{1}\subset [m]$ and $M_{2}\subset [m]$. Assume set $M_{1}$ is an acyclic set of $\mathcal{I}$ and $M_{2}\subset B_{i}$ for all $i\in M_{1}$. Now, if $\mathrm{rank}(\boldsymbol{H}^{M_{1}\cup M_{2}})\leq kt$ where $|M_{1}| < k \leq |M_{1}|+|M_{2}|$, then we must have $\mathrm{rank}(\boldsymbol{H}^{M_{2}})\leq (|M_{1}| - k) t$.
\end{lem}

\begin{proof}
    Let $M_{1}=\{i_{1}, \dots, i_{|M_{1}|}\}$. Applying the decoding condition \eqref{eq:dec-cond} to each $i_{1}, \dots, i_{|M_{1}|}$ will result in
    \begin{align}
        &\mathrm{rank}(\boldsymbol{H}^{\{i_{1}, \dots, i_{|M_{1}|}\}\cup M_{2}})
        = 
        \nonumber
        \\
        &\mathrm{rank}(\boldsymbol{H}^{\{i_{1}\}}) + 
        \mathrm{rank}(\boldsymbol{H}^{\{i_{2}, \dots, i_{|M_{1}|}\}\cup M_{2}})=
        \nonumber
        \\
        &\vdots
        \nonumber
        \\
        &\mathrm{rank}(\boldsymbol{H}^{\{i_{1}\}}) + \dots \mathrm{rank}(\boldsymbol{H}^{\{i_{|M_{1}|}\}}) + \mathrm{rank}(\boldsymbol{H}^{M_{2}}) =
        \nonumber
        \\
        & |M_{1}|t + \mathrm{rank}(\boldsymbol{H}^{M_{2}}).
    \end{align}
    Now, if $\mathrm{rank}(\boldsymbol{H}^{M_{1}\cup M_{2}})\leq kt$, then we will have $\mathrm{rank}(\boldsymbol{H}^{M_{2}})\leq (k - |M_{1}|) t$.
\end{proof}

\subsection{The Class $p$-Fano Index Coding Instance} \label{sub:index-coding-fano}
\begin{defn}[$\mathcal{I}_{F}(p)$: Class $p$-Fano Index Coding Instance]
Consider the following class of index coding instances $\mathcal{I}_{F}(p)=\{B_{i}, i\in [m]\}$,
\begin{equation}
    [m]=[n] \cup_{l\in [p+1]} Z_{l} \cup Z^{\prime} \cup_{l\in [p+1]} Z_{l}^{\prime \prime},
\end{equation}
where $n=2p+3$ and
\begin{align}
    &Z_{l}= \{z(l,j), j\in [p+1]\cup \{l\}\}, \
    \nonumber
    \\
    &z(l,j)=(p+1)(l+1)+1+j,\ \ l,j\in [p+1],
    \\
    &Z^{\prime}= \{z^{\prime}(j), j\in [p+1]\}, \ \ \ \ \ \ \ \ \ \ \ 
    \nonumber
    \\
    &z^{\prime}(j)= p^{2}+4p+4+j, \ \ j\in [p+1],
    \\
    &Z_{l}^{\prime\prime}= \{z^{\prime\prime}(l,j), j\in [p-2]\},\ p>2,
    \nonumber
    \\
    &z^{\prime\prime}(l,j)=p^{2}+(l+4)p+7-2i+j, \ \ l\in [p+1], j\in [p-2].
\end{align}
The interfering message sets $B_{i}, i\in [n]$ are as follows.
\begin{align}
    &B_{i}=([p+1]\backslash \{i\})\cup \{n-i\}\cup_{l\in [p+1]} Z_{l}\backslash\{z(l,i)\}, i\in [p+1],  
    \nonumber
    \\
    &B_{i}=[p+2:2p+2]\backslash \{i,i+1\}, \ i\in [p+2:2p],
    \nonumber
    \\
    &B_{i}=[p+2:2p+2]\backslash \{i,1\}, \ i = 2p + 1,
    \nonumber
    \\
    &B_{2p+2}=[p+2:2p+2]\backslash \{2p+2,p+2\},
    \nonumber
    \\
    &B_{n} = [p+2:2p+2];
    \label{eq:Fano-instane-b1n}
\end{align}
The interfering message sets $B_{i}, i\in Z_{l}, l\in [p+1]$ are as follows
\begin{align}
    &B_{z(l,l)}= \emptyset, \ \ \  l\in [p+1],
    \\
    &B_{z(l,j)}= (Z_{l}\backslash \{z(l,l), z(l,j), z(l, j+1)\}) \cup \{n-l\},
    \nonumber
    \\  
    &j \in [p]\backslash \{l,l-1\}, l \in [P+1],
    \\
    &B_{z(l,l-1)}=(Z_{l}\backslash \{z(l,l), z(l,l-1), z(l, l+2)\}) \cup \{n-l\},
    \nonumber
    \\
    &l\in [p]\backslash \{1\},
    \\
    &B_{z(l,l-1)}=(Z_{l}\backslash \{z(l,l), z(l,l-1), z(l, 1)\}) \cup \{n-l\},
    \nonumber
    \\
    &l=p+1,
    \\
    &B_{z(l,p+1)}=(Z_{l}\backslash \{z(l,l), z(l,p+1), z(l, 2)\}) \cup \{n-l\}, \
    \nonumber
    \\
    &l=1,
    \\
    &B_{z(l,p+1)}=(Z_{l}\backslash \{z(l,l), z(l,p+1), z(l, 1)\}) \cup \{n-l\}, 
    \nonumber
    \\
    &l\in [p+1]\backslash \{1,p+1\}.
\end{align}
The interfering message sets $B_{i}, i\in Z^{\prime}$ are as follows
\begin{align}
    B_{z^{\prime}(l)}& =\{n-l, n, z(l,l)\}\cup Z_{l}^{\prime\prime},\ \ \ l\in [p+1].
\end{align}
The interfering message sets $B_{i}, i\in Z_{l}^{\prime\prime}, l\in [p+1]$ are as follows
\begin{align}
    B_{z^{\prime\prime}(l,j)}&= \{n-l, n, z(l,l) \}\cup \{z^{\prime}(l)\}\cup \big (Z_{l}^{\prime\prime}\backslash \{z^{\prime\prime}(l,j)\}\big ),
    \\
    &l\in [p+1], j\in [p-2].
    \nonumber
\end{align}
We refer to such $\mathcal{I}_{F}(p)$ as the class $p$-Fano index coding instance. It can be seen that the total number of users is 
\begin{align}
    m&= n + \sum_{l\in [p+1]} |Z_{l}| + |Z_{l}^{\prime}| + \sum_{l\in [p+1]} |Z_{l}^{\prime\prime}|
    \nonumber
    \\
    &= (2p+3) + (p+1)(p+1) + (p+1) + (p+1)(p-2)
    \nonumber
    \\
    &=2p^{2}+4p+3.
\end{align}
\end{defn}

    


\begin{exmp}[$\mathcal{I}_{F}(2)$] \label{2-Fano index coding instance}
Consider the class $2$-Fano index coding instance $\mathcal{I}_{F}(2)=\{B_{i}, i\in [m]\}$, which is characterized as follows
\begin{align}
     &Z_{1}=\{z(1,1)=8,z(1,2)=9,z(1,3)=10\},
     \nonumber
     \\
     &Z_{2}=\{z(2,1)=11,z(2,2)=12,z(2,3)=13\},
     \nonumber
     \\
     &Z_{3}=\{z(3,1)=14,z(3,2)=15,z(3,3)=16\},
     \nonumber
     \\
     &Z^{\prime}=\{z^{\prime}(1)=17 ,z^{\prime}(2)= 18,z^{\prime}(3)= 19\},
    \nonumber
    \\
    &m= n + |Z_{1}| +|Z_{2}| +|Z_{3}| + |Z^{\prime}|= 7 + 3 +3+3+3=19.
    \nonumber
\end{align}
\begin{align}
    &B_{1}=([3]\backslash \{1\})\cup \{6\}\cup [8:16]\backslash \{8,11,14\},
    \nonumber
    \\
    &B_{2}=([3]\backslash \{2\})\cup \{5\}\cup [8:19]\backslash \{9,12,15\},
    \nonumber
    \\
    &B_{3}=([3]\backslash \{3\})\cup \{4\}\cup [8:19]\backslash \{10,13,16\},
    \nonumber
    \\
    &B_{4}=\{6\},\ \ B_{5}=\{4\},\ \ B_{2p+2=6}=\{5\},
    \nonumber
    \\
    &B_{n=7}=\{4,5,6\},
    \nonumber
    \\
    &B_{z(1,1)=8}=\emptyset,\ \ B_{z(1,2)=9}=\{6\}, \ \ B_{z(1,3)=10}=\{6\},
    \nonumber
    \\
    &B_{z(2,1)=11}=\{5\},\ \ B_{z(2,2)=12}=\emptyset, \ \ B_{z(2,3)=13}=\{5\},
    \nonumber
    \\
    &B_{z(3,1)=14}=\{4\},\ \ B_{z(3,2)=15}=\{4\}, \ \ B_{z(3,3)=16}=\emptyset
    \nonumber
    \\
    &B_{z^{\prime}(1)=17}=\{6,7,8\},\ \ B_{z^{\prime}(2)=18}=\{5,7,12\}, 
    \nonumber
    \\
    &B_{z^{\prime}(3)=19}=\{4,7,16\}.
    \nonumber
\end{align}
\end{exmp}

It is worth mentioning that applying the mapping methods discussed in \cite{Effros2015} and \cite{Maleki2014} for the 2-Fano matroid results in an index coding instance consisting of more than 200 users, while the 2-Fano index coding instance in Example \ref{2-Fano index coding instance} is of a size 19.

\begin{exmp}[$\mathcal{I}_{F}(3)$] \label{3-Fano index coding instance}
Consider the class $3$-Fano index coding instance $\mathcal{I}_{F}(3)=\{B_{i}, i\in [m]\}$, which is characterized as follows
\begin{align}
     &Z_{1}=\{z(1,1)=10,z(1,2)=11,z(1,3)=12, z(1,4)=13\},
     \nonumber
     \\
     &Z_{2}=\{z(2,1)=14,z(2,2)=15,z(2,3)=16, z(2,4)=17\},
     \nonumber
     \\
     &Z_{3}=\{z(3,1)=18,z(3,2)=19,z(3,3)=20, z(3,4)=21\},
     \nonumber
     \\
     &Z_{4}=\{z(4,1)=22,z(4,2)=23,z(4,3)=24, z(4,4)=25\},
     \nonumber
     \\
     &Z^{\prime}=\{z^{\prime}(1)=26 ,z^{\prime}(2)= 27,z^{\prime}(3)= 28,z^{\prime}(4)= 29\},
    \nonumber
    \\
    &Z_{1}^{\prime\prime}=\{z^{\prime\prime}(1,1)=30\},
    \nonumber
    \\
    &Z_{2}^{\prime\prime}=\{z^{\prime\prime}(2,1)=31\},
    \nonumber
    \\
    &Z_{3}^{\prime\prime}=\{z^{\prime\prime}(3,1)=32\},
    \nonumber
    \\
    &Z_{4}^{\prime\prime}=\{z^{\prime\prime}(4,1)=33\},
    \nonumber
    \\
    &m= n + |Z_{1}| +|Z_{2}| +|Z_{3}| + |Z_{4}| + |Z^{\prime}| + |Z_{1}^{\prime\prime}| +|Z_{2}^{\prime\prime}| +
    \nonumber
    \\
    &|Z_{3}^{\prime\prime}| + |Z_{4}^{\prime\prime}| 
    \nonumber
    \\
    &= 9 + 4 + 4 + 4 + 4 + 4 + 1 + 1 + 1 + 1=33.
    \nonumber
\end{align}
\begin{align}
    &B_{1}=([4]\backslash \{1\})\cup \{8\}\cup [10:25]\backslash \{10,14,18,22\},
    \nonumber
    \\
    &B_{2}=([4]\backslash \{2\})\cup \{7\}\cup [10:25]\backslash \{11,15,19,23\},
    \nonumber
    \\
    &B_{3}=([4]\backslash \{3\})\cup \{6\}\cup [10:25]\backslash \{12,16,20,24\},
    \nonumber
    \\
    &B_{4}=([4]\backslash \{4\})\cup \{5\}\cup [10:25]\backslash \{13,17,21,25\},
    \nonumber
    \\
    &B_{5}=\{7,8\}, B_{6}=\{5,8\}, B_{7}=\{5,6\}, B_{2p+2=8}=\{6,7\},
    \nonumber
    \\
    &B_{n=9}=\{5,6,7,8\},
    \nonumber
    \\
    &B_{z(1,1)=10}=\emptyset,\ \ B_{z(1,2)=11}=\{13\}\cup \{8\}, 
    \nonumber
    \\
    &B_{z(1,3)=12}=\{11\}\cup \{8\}, \ \ B_{z(1,4)=13}=\{12\}\cup \{8\}
    \nonumber
    \\
    &B_{z(2,1)=14}=\{17\}\cup \{7\},\ \ B_{z(2,2)=15}=\emptyset, \
    \nonumber
    \\
    &B_{z(2,3)=16}=\{14\}\cup \{7\}, \ \ B_{z(2,4)=17}=\{16\}\cup \{7\}
    \nonumber
    \\
    &B_{z(3,1)=18}=\{21\}\cup \{6\},\ \ B_{z(3,2)=19}=\{18\}\cup \{6\}, \
    \nonumber
    \\
    &B_{z(3,3)=20}=\emptyset, \ \ B_{z(3,4)=21}=\{19\}\cup \{6\}
    \nonumber
    \\
    &B_{z(4,1)=22}=\{24\}\cup \{5\},\ \ B_{z(4,2)=23}=\{22\}\cup \{5\}, 
    \nonumber
    \\
    &B_{z(4,3)=24}=\{23\}\cup \{5\}, \ \ B_{z(4,4)=25}=\emptyset
    \nonumber
    \\
    &B_{z^{\prime}(1)=26}=\{8,9,10,30\}, B_{z^{\prime}(2)=27}=\{7,9,15,31\},  
    \nonumber
    \\
    &B_{z^{\prime}(3)=28}=\{6,9,20,32\},  
    B_{z^{\prime}(4)=29}=\{5,9,25,33\}.
    \nonumber
    \\
    &B_{z^{\prime\prime}(1,1)=30}=\{8,9,10,26\}, B_{z^{\prime\prime}(2,1)=31}=\{7,9,15,27\}, 
    \nonumber
    \\
    &B_{z^{\prime\prime}(3,1)=32}=\{6,9,20,28\},  
    B_{z^{\prime\prime}(4,1)=33}=\{5,9,25,29\}.
    \nonumber
\end{align}
\end{exmp}

It is worth mentioning that applying the mapping methods discussed in \cite{Effros2015} and \cite{Maleki2014} for the 3-Fano matroid results in an index coding instance consisting of more than 1000 users, while the 3-Fano index coding instance in Example \ref{3-Fano index coding instance} is of a size of 33. 
Moreover, in \cite{Arman-4-3} an index coding instance based on 3-Fano matroid instance was built with the size of 29 users. That index coding instance is closely related to the 3-Fano index coding instance in Example \ref{3-Fano index coding instance}.

\begin{thm}
$\lambda_{q}(\mathcal{I}_{F}(p))=\beta (\mathcal{I}_{F}(p))=p+1$ iff the field $\mathbb{F}_{q}$ does have characteristic $p$. In other words, the necessary and sufficient condition for linear index coding to be optimal for $\mathcal{I}_{F}(p)$ is that the field $\mathbb{F}_{q}$ has characteristic $p$.
\end{thm}
\begin{proof}
The proof can be concluded from Propositions  \ref{prop:Fano-necessary-index} and \ref{prop:I-fano-necessary}.
\end{proof}
We first provide Remark \ref{rem: decoding-user-n} which will be used in the proof of Proposition \ref{prop:Fano-necessary-index}.
\begin{rem} \label{rem: decoding-user-n}
    Consider the following matrix:
    \begin{align}
    \begin{bmatrix}
    1       & 1      & 1      & \hdots & 1      & 1      & 0      & 1\\
    1       & 1      & 1      & \hdots & 1      & 0      & 1      & 1 \\
    1       & 1      & 1      & \hdots & 0      & 1      & 1      & 1 \\
    \vdots  & \vdots & \vdots & \ddots & \vdots & \vdots & \vdots & 1 \\
    1       & 1      & 0      & \hdots & 1      & 1      & 1      & 1\\
    1       & 0      & 1      & \hdots & 1      & 1      & 1      & 1\\
    0       & 1      & 1      & \hdots & 1      & 1      & 1      
    \end{bmatrix}\in \mathbb{F}_{q}^{(p+1)\times (p+2)}.
\end{align}
By running the Guassian elimination technique, we get
\begin{align}
    \begin{bmatrix}
    \textcolor{red}{{\large \textcircled{\small 1}}}       & 1                                                      & 1                                                      & \hdots & 1                                                    & 1                                                     & 0                              &  1
    \\
    0                                                      & 0                                                      & 0                                                      & \hdots & 0                                                    & \textcolor{red}{{\large \textcircled{\small -1}}}     & 1                              &  0      
    \\
    0                                                      & 0                                                      & 0                                                     & \hdots & \textcolor{red}{{\large \textcircled{\small -1}}}     & 0                                                     & 1                              &  0      
    \\
    \vdots                                                 & \vdots                                                 & \vdots                                                & \ddots & \vdots                                                & \vdots                                                & \vdots                         &  0
    \\
    0                                                      & 0                                                      & \textcolor{red}{{\large \textcircled{\small -1}}}     & \hdots & 0                                                     & 0                                                     & 1                              &  0      
    \\
    0                                                      & \textcolor{red}{{\large \textcircled{\small -1}}}      & 0                                                     & \hdots & 0                                                     & 0                                                     & 1                              &  0      
    \\
    0                                                      & 0                                                      & 0                                                     & \hdots & 0                                                     & 0                                                     &   \textcolor{red}{1+(p-1)}     &  \textcolor{red}{{\large \textcircled{\small 1}}}    
    \end{bmatrix}.
\end{align}
\end{rem}
Since $1+(p-1)=0$ over fields $\mathbb{F}_{q}$ with characteristic $p$, then the last column is linearly independent of the column space of the other columns.
\begin{prop} \label{prop:Fano-necessary-index}
For the class $p$-Fano index coding instance $\mathcal{I}_{F}(q)$, there exists an optimal scalar linear code $(t=1)$ over the field $F_{q}$ with characteristic $p$.
\end{prop}

\begin{proof}
First, it can be observed that set $[p+1]$ is an independent set of $\mathcal{I}_{F}(p)$. Thus, $\beta (\mathcal{I}_{F}(p))\geq p+1$.
\\
Now, consider the scalar encoding matrix $\boldsymbol{H}_{p}\in \mathbb{F}_{p}^{(p+1)\times m}$, characterized by its columns as follows (shown in Figure \ref{fig:H-class-fano-nonfano})
\begin{itemize}

    \item Users $u_{i}, i\in [n]$:
    \begin{align}
        \boldsymbol{H}_{p}^{\{i\}}=
    \left \{
    \begin{array}{ccc}
     &\boldsymbol{I}_{p+1}^{\{i\}},\ \ \ \ \ \ \ \ \ \ \ \ \ \ \ \ \  \ i\in [p+1], \ \ \ \ \ \ \ \ \ \ \ 
     \\
     \\
     &\sum_{j\in [p+1]\backslash \{n-i\}} \boldsymbol{I}_{p+1}^{\{j\}}, \ i\in [p+2:2p+2],
     \\
     \\
     &\sum_{j\in [p+1]} \boldsymbol{I}_{p+1}^{\{j\}}, \ \ \ \ \ \ i=n=2p+3, \ \ \ \ \ \ 
    \end{array}
    \right.
    \end{align} 
    
    \item $\boldsymbol{H}_{p}^{\{z(l,j)\}}=\boldsymbol{I}_{p+1}^{\{j\}}, \ l,j\in [p+1]$,

    \item
    \begin{align}
    \boldsymbol{H}_{p}^{\{z^{\prime}(l)\}}=
    \left \{
    \begin{array}{cc}
     &\boldsymbol{I}_{p+1}^{\{l+1\}},\ \ \ \  l\in [p], \ \ \ \ \ \ \ \ \ \ \
     \\
     \\
     &\boldsymbol{I}_{p+1}^{\{1\}},\ \ \ \ \ l = p + 1, \ \ \ \ \ \ 
    \end{array}
  \right.
    \end{align} 
    \item 

    \begin{align}
    \boldsymbol{H}_{p}^{\{z^{\prime\prime}(l,j)\}}=
    \left \{
    \begin{array}{cc}
     &\boldsymbol{I}_{p+1}^{\{l+j+1\}},\ \ l+j\leq p, \ \ \ \ \ \ \ \ \ \ \ \ \ \ \ \ \ \ \
     \\
     \\
     &\boldsymbol{I}_{p+1}^{\{l+j-p\}},\ \ p+1\leq l+j\leq 2p-1,\ \
    \end{array}
    \right.
    \end{align} 
\end{itemize}

\begin{figure*}
\centering
\subfloat[$\boldsymbol{H}_{p}^{[n=2p+3]}$]
{
$
\begin{blockarray}{cccccccccccc}
              & 
              \textcolor{blue}{{\tiny 1}}        & 
              \textcolor{blue}{{\tiny 2}}        & 
              \textcolor{blue}{{\tiny \hdots}}   & 
              \textcolor{blue}{{\tiny p}}        &
              \textcolor{blue}{{\tiny p+1}}      &
              \textcolor{blue}{{\tiny p+2}}      & 
              \textcolor{blue}{{\tiny p+3}}      & 
              \textcolor{blue}{{\tiny \hdots}}   & 
              \textcolor{blue}{{\tiny 2p+1}}     &
              \textcolor{blue}{{\tiny 2p+2}}     &
              \textcolor{blue}{{\tiny 2p+3}}  \\ 
\begin{block}{c[ccccc|ccccc|c]}
              \textcolor{blue}{{\tiny 1}} & 
              1 & 0 & \hdots & 0 & 0 &
              1 & 1 & \hdots & 1 & 0 &
              1
              \\
              \textcolor{blue}{{\tiny 2}} & 
              0 & 1 & \hdots & 0 & 0 &
              1 & 1 & \hdots & 0 & 1 &
              1
              \\
              \textcolor{blue}{{\tiny \vdots}} & 
              \vdots & \vdots & \ddots & \vdots & \vdots & 
              \vdots & \vdots & \ddots & \vdots & \vdots &
              \vdots &
              \\
              \textcolor{blue}{{\tiny p}} & 
              0 & 0 & \hdots & 1 & 0 &
              1 & 0 & \hdots & 1 & 1 &
              1 &
              \\
              \textcolor{blue}{{\tiny p+1}} & 
              0 & 0 & \hdots & 0 & 1 & 
              0 & 1 & \hdots & 1 & 1 & 
              1 
              \\
                \end{block}
\end{blockarray} \label{fig:H-[n=2p+3]}
$
}
\\
\subfloat[$\boldsymbol{H}_{p}^{Z_{l}}=\boldsymbol{I}_{p+1}$]
{
$
\begin{blockarray}{cccccc}
              & 
              \textcolor{blue}{{\tiny z(l,1)}}        & 
              \textcolor{blue}{{\tiny z(l,2)}}        & 
              \textcolor{blue}{{\tiny \hdots}}   & 
              \textcolor{blue}{{\tiny z(l,p)}}        &
              \textcolor{blue}{{\tiny z(l,p+1)}}     \\ 
\begin{block}{c[ccccc]}
              \textcolor{blue}{{\tiny 1}} & 
              1 & 0 & \hdots & 0 & 0
              \\
              \textcolor{blue}{{\tiny 2}} & 
              0 & 1 & \hdots & 0 & 0 
              \\
              \textcolor{blue}{{\tiny \vdots}} & 
              \vdots & \vdots & \ddots & \vdots & \vdots 
              \\
              \textcolor{blue}{{\tiny p}} & 
              0 & 0 & \hdots & 1 & 0 
              \\
              \textcolor{blue}{{\tiny p+1}} & 
              0 & 0 & \hdots & 0 & 1 
              \\
            \end{block}
\end{blockarray} \label{fig:H-Z_l}
$
}
\\
\subfloat[$\boldsymbol{H}_{p}^{Z^{\prime}}$]
{
$
\begin{blockarray}{ccccccc}
              & 
              \textcolor{blue}{{\tiny z^{\prime}(1)}}        & 
              \textcolor{blue}{{\tiny z^{\prime}(2)}}        & 
              \textcolor{blue}{{\tiny \hdots}}   & 
              \textcolor{blue}{{\tiny z^{\prime}(p-1)}}        &
              \textcolor{blue}{{\tiny z^{\prime}(p)}}        &
              \textcolor{blue}{{\tiny z^{\prime}(p+1)}}     \\ 
\begin{block}{c[cccccc]}
              \textcolor{blue}{{\tiny 1}} & 
              0 & 0 & \hdots & 0 & 0 &1
              \\
              \textcolor{blue}{{\tiny 2}} & 
              1 & 0 & \hdots & 0 & 0 & 0 
              \\
              \textcolor{blue}{{\tiny 3}} & 
              0 & 1 & \hdots & 0 & 0 & 0 
              \\
              \textcolor{blue}{{\tiny \vdots}} & 
              \vdots & \vdots & \ddots & \vdots & \vdots & \vdots 
              \\
              \textcolor{blue}{{\tiny p}} & 
              0 & 0 & \hdots & 1 & 0 & 0 
              \\
              \textcolor{blue}{{\tiny p+1}} & 
              0 & 0 & \hdots & 0 & 1 & 0 
              \\
             \end{block}
\end{blockarray} \label{fig:H-Z_prime}
$
}
\\
\subfloat[$\boldsymbol{H}_{p}^{Z_{1}^{\prime\prime}}$]
{
$
\begin{blockarray}{cccccc}
              & 
              \textcolor{blue}{{\tiny z^{\prime\prime}(1,1)}} 
              & 
              \textcolor{blue}{{\tiny z^{\prime\prime}(1,2)}} 
              & 
              \textcolor{blue}{{\tiny \hdots}}   
              & 
              \textcolor{blue}{{\tiny z^{\prime\prime}(1,p-2)}}
              &
              \textcolor{blue}{{\tiny z^{\prime\prime}(1,p-1)}}  \\ 
\begin{block}{c[ccccc]}
              \textcolor{blue}{{\tiny 1}} & 
              0 & 0 & \hdots & 0 & 0 
              \\
              \textcolor{blue}{{\tiny 2}} & 
              0 & 0 & \hdots & 0 & 0  
              \\
              \textcolor{blue}{{\tiny 3}} & 
              1 & 0 & \hdots & 0 & 0 
              \\
              \textcolor{blue}{{\tiny 4}} & 
              0 & 1 & \hdots & 0 & 0  
              \\
              \textcolor{blue}{{\tiny \vdots}} & 
              \vdots & \vdots & \ddots & \vdots & \vdots
              \\
              \textcolor{blue}{{\tiny p}} & 
              0 & 0 & \hdots & 1 & 0
              \\
              \textcolor{blue}{{\tiny p+1}} & 
              0 & 0 & \hdots & 0 & 1
              \\
             \end{block}
\end{blockarray} \label{fig:H-Z_l-prime-prime}
$
} 
\\
\subfloat
{
\vdots
}
\\
\subfloat[$\boldsymbol{H}_{p}^{Z_{p+1}^{\prime\prime}}$]
{
$
\begin{blockarray}{cccccc}
              & 
              \textcolor{blue}{{\tiny z^{\prime\prime}(p+1,1)}}        
              & 
              \textcolor{blue}{{\tiny z^{\prime\prime}(p+1,2)}}        
              & 
              \textcolor{blue}{{\tiny \hdots}}   
              &
              \textcolor{blue}{{\tiny z^{\prime\prime}(p+1,p-2)}} 
              &
              \textcolor{blue}{{\tiny z^{\prime\prime}(p+1,p-1)}}     \\ 
\begin{block}{c[ccccc]}
              \textcolor{blue}{{\tiny 1}} & 
              0 & 0 & \hdots & 0 & 0
              \\
              \textcolor{blue}{{\tiny 2}} & 
              1 & 0 & \hdots & 0 & 0
              \\
              \textcolor{blue}{{\tiny 3}} & 
              0 & 1 & \hdots & 0 & 0
              \\
              \textcolor{blue}{{\tiny \vdots}} & 
              \vdots & \vdots & \ddots & \vdots & \vdots
              \\
              \textcolor{blue}{{\tiny p-1}} & 
              0 & 0 & \hdots & 1 & 0
              \\
              \textcolor{blue}{{\tiny p}} & 
              0 & 0 & \hdots & 0 & 1
              \\
              \textcolor{blue}{{\tiny p+1}} & 
              0 & 0 & \hdots & 0 & 0
              \\
             \end{block}
\end{blockarray} \label{fig:H-Z_(p+1)-prime-prime}
$
}
    \caption{$\boldsymbol{H}_{p}\in \mathbb{F}_{q}^{(p+1)\times m}$: If $\mathbb{F}_{q}$ does have characteristic $p$, then $\boldsymbol{H}_{p}$ is an encoding matrix for the class $p$-Fano index coding instance $\mathcal{I}_{F}(p)$, and if $\mathbb{F}_{q}$ does have any characteristic other than characteristic $p$, then $\boldsymbol{H}_{p}$  is an encoding matrix for the class $p$-non-Fano index coding instance $\mathcal{I}_{nF}(p)$.}
    \label{fig:H-class-fano-nonfano}
\end{figure*}

Now, we prove that this encoding matrix can satisfy all the users.
\begin{itemize}
    \item For users $u_{i}, i\in [p+1]$ with the interfering message set $B_{i}=([p+1]\backslash \{i\})\cup \{n-i\}\cup_{j\in [p+1]} Z_{l}\backslash\{z(l,i)\}$, we have
    \begin{itemize}
        \item $\boldsymbol{H}_{p}^{[p+1]\backslash \{i\}}=\boldsymbol{I}_{p+1}^{[p+1]\backslash \{i\}}$.
        \item $\boldsymbol{H}_{p}^{\{n-i\}}=\sum_{j\in [p+1]\backslash \{i\}} \boldsymbol{I}_{p+1}^{\{i\}}$.
        \item $\boldsymbol{H}_{p}^{Z_{l}\backslash\{z(l,i)\}}=\boldsymbol{I}_{p+1}^{[p+1]\backslash \{i\}}$.
    \end{itemize}
    It can be checked that  
    \begin{align}
        \boldsymbol{H}_{p}^{B_{i}}\stackrel{\mathrm{rref}}{\equiv}\boldsymbol{I}_{p+1}^{[p+1]\backslash \{i\}}.
        \nonumber
    \end{align}
    Since $\boldsymbol{H}_{p}^{\{i\}}=\boldsymbol{I}_{p+1}^{\{i\}}$, the decoding condition \eqref{eq:dec-cond} is met, and user $u_{i}, i\in [p+1]$ can decode its requested message from the $i$-th transmission $y_{i}$.

    \item According to Lemma \ref{lem:matrix-F-(p+1)-p}, set $[p+2:2p+2]$ is a circuit set of $\boldsymbol{H}_{p}$, which means that each set $[p+2:2p+2]\backslash \{j\}, j\in [p+2:2p+2]$ is an independent set. Now, Since $B_{i}=[p+2:2p+2]\backslash \{i, i+1\}, i\in [p+2:2p+1]$ and $B_{2p+2}=\{p+2:2p+2\}\backslash \{1, 2p+2\}$, the decoding condition in \eqref{eq:dec-cond} will be met for all users $u_{i}, i\in [p+2:2p+2]$.
    \item For user $u_{n}$, according to Remark \ref{rem: decoding-user-n}, column $\boldsymbol{H}_{p}^{\{n\}}$ is linearly independent of the column space of its interfering messages set $\boldsymbol{H}_{p}^{B_{n}=[p+2:2p+2]}$, which satisfies the decoding condition \eqref{eq:dec-cond} for $i=n$.
    \item Since the interfering message set of users $u_{z(l,l)}, l\in [p+1]$ is empty and $\boldsymbol{H}_{p}^{z(l,l)} = \boldsymbol{I}_{p+1}^{\{l\}}$, each user $u_{z(l,l)}$ can decode its demanded message from the $l$-th transmission.
    \item For users $u_{z(l,j)}, l\in [p+1], j\in [p]\backslash \{l-1, l\}$ with the interfering message set $B_{z(l,j)}= (Z_{l}\backslash \{z(l,l), z(l,j), z(l, j + 1)\}) \cup \{n-l\}$, we have
    \begin{itemize}
        \item $\boldsymbol{H}_{p}^{Z_{l}\backslash \{z(l,l), z(l,j), z(l, j+ 1)\})}=\boldsymbol{I}_{p+1}^{[p+1]\backslash \{l, j, j+1\}}$,
        
        \item $\boldsymbol{H}_{p}^{\{n-l\}}=\sum_{i\in [p+1]\backslash \{l\}} \boldsymbol{I}_{p+1}^{\{i\}}$.
    \end{itemize}
    It can be verified that
    \begin{align}
        \boldsymbol{H}_{p}^{B_{z(l,j)}} &=
        \left [\begin{array}{c|c} 
         \boldsymbol{I}_{p+1}^{[p+1]\backslash \{l, j, j+1\}} & \sum_{i\in [p+1]\backslash \{l\}} \boldsymbol{I}_{p+1}^{\{i\}}
         \end{array}
         \right ].
         \nonumber
         \\
        &\stackrel{\mathrm{rref}}{\equiv}
        \left [\begin{array}{c|c} 
         \boldsymbol{I}_{p+1}^{[p+1]\backslash \{l, j, j+1\}} & \boldsymbol{I}_{p+1}^{\{j\}} + \boldsymbol{I}_{p+1}^{\{j+1\}}
         \end{array}
         \right ].
        \nonumber
    \end{align}

    \item For users $u_{z(l,l-1)}, l\in [p]\backslash \{1\}$ with the interfering message set $B_{z(l,j)}= (Z_{l}\backslash \{z(l,l), z(l,l-1), z(l, l + 2)\}) \cup \{n-l\}$, we have
    \begin{itemize}
        \item $\boldsymbol{H}_{p}^{Z_{l}\backslash \{z(l,l), z(l,l-1), z(l, l+ 2)\})}=\boldsymbol{I}_{p+1}^{[p+1]\backslash \{l, l-1, l+2\}}$,
        
        \item $\boldsymbol{H}_{p}^{\{n-l\}}=\sum_{i\in [p+1]\backslash \{l\}} \boldsymbol{I}_{p+1}^{\{i\}}$.
    \end{itemize}
    It can be verified that
    \begin{align}
        \boldsymbol{H}_{p}^{B_{z(l,l-1)}} &=
        \left [\begin{array}{c|c} 
         \boldsymbol{I}_{p+1}^{[p+1]\backslash \{l, l-1, l+2\}} & \sum_{i\in [p+1]\backslash \{l-1\}} \boldsymbol{I}_{p+1}^{\{i\}}
         \end{array}
         \right ].
         \nonumber
         \\
        &\stackrel{\mathrm{rref}}{\equiv}
        \left [\begin{array}{c|c} 
         \boldsymbol{I}_{p+1}^{[p+1]\backslash \{l, l-1, l+2\}} & \boldsymbol{I}_{p+1}^{\{l\}} + \boldsymbol{I}_{p+1}^{\{l+2\}}
         \end{array}
         \right ].
        \nonumber
    \end{align}
     Now, it can be seen that column $\boldsymbol{H}_{p}^{\{z(l,l-1)\}}=\boldsymbol{I}_{p+1}^{\{l-1\}}$ is linearly independent of the column space of $\boldsymbol{H}_{p}^{ B_{z(l,l-1)}}$. Thus, the decoding condition \eqref{eq:dec-cond} will be satisfied for users $u_{z(l,l-1)}, l\in [p]\backslash \{1\}$.

    \item For users $u_{z(l,l-1)}, l=p+1$ with the interfering message set $B_{z(l,l-1)}= (Z_{l}\backslash \{z(l,l), z(l,l-1), z(l, 1)\}) \cup \{n-l\}$, we have 
    \begin{itemize}
        \item $\boldsymbol{H}_{p}^{Z_{l}\backslash \{z(l,l), z(l,l-1), z(l, 1)\})}=\boldsymbol{I}_{p+1}^{[p+1]\backslash \{l, l-1, 1\}}$,
        
        \item $\boldsymbol{H}_{p}^{\{n-l\}}=\sum_{i\in [p+1]\backslash \{l\}} \boldsymbol{I}_{p+1}^{\{i\}}$.
    \end{itemize}
    It can be verified that
    \begin{align}
        \boldsymbol{H}_{p}^{B_{z(l,l-1)}} &=
        \left [\begin{array}{c|c} 
         \boldsymbol{I}_{p+1}^{[p+1]\backslash \{l, l-1, 1\}} & \sum_{i\in [p+1]\backslash \{l\}} \boldsymbol{I}_{p+1}^{\{i\}}
         \end{array}
         \right ]
         \nonumber
         \\
        &\stackrel{\mathrm{rref}}{\equiv}
        \left [\begin{array}{c|c} 
         \boldsymbol{I}_{p+1}^{[p+1]\backslash \{l, l-1, 1\}} & \boldsymbol{I}_{p+1}^{\{l-1\}} + \boldsymbol{I}_{p+1}^{\{1\}}
         \end{array}
         \right ].
        \nonumber
    \end{align}
     Now, it can be seen that column $\boldsymbol{H}_{p}^{\{z(l,l-1)\}}=\boldsymbol{I}_{p+1}^{\{l-1\}}$ is linearly independent of the column space of $\boldsymbol{H}_{p}^{ B_{z(l,l-1)}}$. Thus, the decoding condition \eqref{eq:dec-cond} will be satisfied for users $u_{z(l,l-1)}, l = p+1$.

     \item For users $u_{z(l,p+1)}, l=1$ with the interfering message set $B_{z(l,p+1)}= (Z_{l}\backslash \{z(l,l), z(l,p+1), z(l, 2)\}) \cup \{n-l\}$, we have 
    \begin{itemize}
        \item $\boldsymbol{H}_{p}^{Z_{l}\backslash \{z(l,l), z(l,p+1), z(l, 2)\})}=\boldsymbol{I}_{p+1}^{[p+1]\backslash \{l, p+1, 2\}}$,
        
        \item $\boldsymbol{H}_{p}^{\{n-l\}}=\sum_{i\in [p+1]\backslash \{l\}} \boldsymbol{I}_{p+1}^{\{i\}}$.
    \end{itemize}
    It can be verified that
    \begin{align}
        \boldsymbol{H}_{p}^{B_{z(l,p+1)}} &=
        \left [\begin{array}{c|c} 
         \boldsymbol{I}_{p+1}^{[p+1]\backslash \{l, p+1, 2\}} & \sum_{i\in [p+1]\backslash \{l\}} \boldsymbol{I}_{p+1}^{\{i\}}
         \end{array}
         \right ].
         \nonumber
         \\
        &\stackrel{\mathrm{rref}}{\equiv}
        \left [\begin{array}{c|c} 
         \boldsymbol{I}_{p+1}^{[p+1]\backslash \{l, l-1, 1\}} & \boldsymbol{I}_{p+1}^{\{p+1\}} + \boldsymbol{I}_{p+1}^{\{2\}}
         \end{array}
         \right ].
        \nonumber
    \end{align}
     Now, it can be seen that column $\boldsymbol{H}_{p}^{\{z(l,p+1)\}}=\boldsymbol{I}_{p+1}^{\{p+1\}}$ is linearly independent of the column space of $\boldsymbol{H}_{p}^{ B_{z(l,p+1)}}$. Thus, the decoding condition \eqref{eq:dec-cond} will be satisfied for users $u_{z(l,p+1)}, l = 1$.

     \item For users $u_{z(l,p+1)}, l=[p]\backslash \{1\}$ with the interfering message set $B_{z(l,p+1)}= (Z_{l}\backslash \{z(l,l), z(l,p+1), z(l, 1)\}) \cup \{n-l\}$, we have 
    \begin{itemize}
        \item $\boldsymbol{H}_{p}^{Z_{l}\backslash \{z(l,l), z(l,p+1), z(l, 1)\})}=\boldsymbol{I}_{p+1}^{[p+1]\backslash \{l, p+1, 1\}}$,
        
        \item $\boldsymbol{H}_{p}^{\{n-l\}}=\sum_{i\in [p+1]\backslash \{l\}} \boldsymbol{I}_{p+1}^{\{i\}}$.
    \end{itemize}
    It can be verified that
    \begin{align}
        \boldsymbol{H}_{p}^{B_{z(l,p+1)}} &=
        \left [\begin{array}{c|c} 
         \boldsymbol{I}_{p+1}^{[p+1]\backslash \{l, p+1, 1\}} & \sum_{i\in [p+1]\backslash \{l\}} \boldsymbol{I}_{p+1}^{\{i\}}
         \end{array}
         \right ].
         \nonumber
         \\
        &\stackrel{\mathrm{rref}}{\equiv}
        \left [\begin{array}{c|c} 
         \boldsymbol{I}_{p+1}^{[p+1]\backslash \{l, l-1, 1\}} & \boldsymbol{I}_{p+1}^{\{p+1\}} + \boldsymbol{I}_{p+1}^{\{1\}}
         \end{array}
         \right ].
        \nonumber
    \end{align}
     Now, it can be seen that column $\boldsymbol{H}_{p}^{\{z(l,p+1)\}}=\boldsymbol{I}_{p+1}^{\{p+1\}}$ is linearly independent of the column space of $\boldsymbol{H}_{p}^{ B_{z(l,p+1)}}$. Thus, the decoding condition \eqref{eq:dec-cond} will be satisfied for users $u_{z(l,p+1)}, l \in [p]\backslash \{1\}$.

     \item For users $u_{z^{\prime}(l)}, l\in [p+1]$ with the interfering message set $B_{z^{\prime}(l)}=\{z(l,l), n-l, n\}\cup Z_{l}^{\prime\prime}$, we have
     \begin{itemize}
         \item It can be verified that
         \begin{align}
         &\boldsymbol{H}_{p}^{\{z(l,l), n-l, n\}}
         \nonumber
         \\
         &=\left [\begin{array}{c|c|c} 
         \boldsymbol{I}_{p+1}^{\{l\}} & \sum_{j\in [p+1]\backslash \{l\}} \boldsymbol{I}_{p+1}^{\{j\}}  & \sum_{j\in [p+1]} \boldsymbol{I}_{p+1}^{\{j\}}
         \end{array}
         \right ]
         \nonumber
         \\
         &\stackrel{\mathrm{rref}}{\equiv}
         \left [\begin{array}{c|c} 
         \boldsymbol{I}_{p+1}^{\{l\}} & \sum_{j\in [p+1]\backslash \{l\}} \boldsymbol{I}_{p+1}^{\{j\}}
         \end{array}
         \right ],
         \end{align}
         \item It can be observed that 
         \begin{align}
        \boldsymbol{H}_{p}^{Z_{l}^{\prime\prime}}=
        \left \{
        \begin{array}{cc}
         &\boldsymbol{I}_{p+1}^{[p+1]\backslash \{l, l+1\}},\ \ \ \ \ \  l\in [p], \ \ \ \ \ \ \ \ \ \ \ \
         \\
         \\
         &\boldsymbol{I}_{p+1}^{[p+1]\backslash \{l, 1\}},\ \ \ \ \ \ \ l = p + 1, \ \ \ \ \ \ \ \
        \end{array}
      \right.
        \end{align} 
     \end{itemize}
     It can be seen that column 
    \begin{align}
    \boldsymbol{H}_{p}^{\{z^{\prime}(l)\}}=
    \left \{
    \begin{array}{cc}
     &\boldsymbol{I}_{p+1}^{\{l+1\}},\ \ \ \  l\in [p],\ \ \ \ \ \ \ \ \
     \\
     \\
     &\boldsymbol{I}_{p+1}^{\{1\}},\ \ \ \ \ \ \ l = p + 1, \ \ \ \ \ \ 
    \end{array}
    \right.
    \end{align} 
     will be linearly independent of the column space of $\boldsymbol{H}_{p}^{B_{z^{\prime}(l)}}$. Thus, the decoding condition \eqref{eq:dec-cond} will be satisfied for users $u_{z^{\prime}(l)}, l\in [p+1]$.
     
     \item For users $u_{z^{\prime\prime}(l,j)}, l\in [p+1], j\in [p-2]$ with the interfering message set $B_{z^{\prime\prime}(l,j)}= \{z(l,l), n-l, n\}\cup \big (Z_{l}^{\prime\prime}\backslash \{z^{\prime\prime}(l,j)\}\big )$, we have
     \begin{itemize}
          \item It can be verified that
         \begin{align}
         &\boldsymbol{H}_{p}^{\{z(l,l), n-l, n\}}
         \nonumber
         \\
         &=\left [\begin{array}{c|c|c} 
         \boldsymbol{I}_{p+1}^{\{l\}} & \sum_{j\in [p+1]\backslash \{l\}} \boldsymbol{I}_{p+1}^{\{j\}}  & \sum_{j\in [p+1]} \boldsymbol{I}_{p+1}^{\{j\}}
         \end{array}
         \right ]
         \nonumber
         \\
         &\stackrel{\mathrm{rref}}{\equiv}
         \left [\begin{array}{c|c} 
         \boldsymbol{I}_{p+1}^{\{l\}} & \sum_{j\in [p+1]\backslash \{l\}} \boldsymbol{I}_{p+1}^{\{j\}}
         \end{array}
         \right ],
         \end{align}
         
         \item 
         
         $\boldsymbol{H}_{p}^{Z_{l}^{\prime\prime}\backslash \{z^{\prime\prime}(l,j)\}}=\boldsymbol{I}_{p+1}^{[p+1]\backslash \{l, l+1, l+j+1\}}$.
     \end{itemize}
      It can be seen that column

               \begin{align}
            \boldsymbol{H}_{p}^{\{z^{\prime\prime}(l,j)\}}=
            \left \{
            \begin{array}{cc}
             &\boldsymbol{I}_{p+1}^{\{l+j+1\}},\ \ l+j\leq p, \ \ \ \ \ \ \ \ \ \ \ \ \ \ \ \ \ \
             \\
             \\
             &\boldsymbol{I}_{p+1}^{\{l+j-p\}},\ \ p+1\leq l+j\leq 2p-1,\ \
            \end{array}
            \right.
        \end{align} 
      
      will be linearly independent of the column space of $\boldsymbol{H}_{p}^{B_{z^{\prime\prime}(l,j)}}$. Thus, the decoding condition \eqref{eq:dec-cond} will be satisfied for users $u_{z^{\prime\prime}(l,j)}, l\in [p+1], j\in [p-2]$.
\end{itemize}
\end{proof}

\begin{exmp}
Figure \ref{fig:H-p-2} depicts the encoding matrix $\boldsymbol{H}_{p=2}\in \mathbb{F}_{q}^{3\times 19}$ which is optimal for the class $2$-Fano index coding instance $\mathcal{I}_{nF}(2)$ if field $\mathbb{F}_{q}$ does have characteristic two.
    \begin{figure*}
\centering
\begin{blockarray}{cccccccccccccccccccc}
              & 
              \textcolor{blue}{{\tiny 1}}\hspace{-1ex} & 
              \textcolor{blue}{{\tiny 2}}\hspace{-1ex} & 
              \textcolor{blue}{{\tiny 3}}\hspace{-1ex} & 
              \textcolor{blue}{{\tiny 4}}\hspace{-1ex} & 
              \textcolor{blue}{{\tiny 5}}\hspace{-1ex} & 
              \textcolor{blue}{{\tiny 6}}\hspace{-1ex} &
              \textcolor{blue}{{\tiny 7}}\hspace{-1ex} &
              \textcolor{blue}{{\tiny 8}}\hspace{-1ex} & 
              \textcolor{blue}{{\tiny 9}}\hspace{-1ex} & 
              \textcolor{blue}{{\tiny 10}}\hspace{-1ex} &
              \textcolor{blue}{{\tiny 11}}\hspace{-1ex} &
              \textcolor{blue}{{\tiny 12}}\hspace{-1ex} &
              \textcolor{blue}{{\tiny 13}}\hspace{-1ex} & 
              \textcolor{blue}{{\tiny 14}}\hspace{-1ex} &
              \textcolor{blue}{{\tiny 15}}\hspace{-1ex} &
              \textcolor{blue}{{\tiny 16}}\hspace{-1ex} &
              \textcolor{blue}{{\tiny 17}}\hspace{-1ex} & 
              \textcolor{blue}{{\tiny 18}}\hspace{-1ex} &
              \textcolor{blue}{{\tiny 19}}
            \\
\begin{block}{c[ccc|ccc|c|ccc|ccc|ccc|ccc]}
               \textcolor{blue}{{\tiny 1}} & 
               1\hspace{-1ex} & 0\hspace{-1ex} & 0 & 
               1\hspace{-1ex} & 1\hspace{-1ex} & 0 & 
               1 & 
               1\hspace{-1ex} & 0\hspace{-1ex} & 0 & 
               1\hspace{-1ex} & 0\hspace{-1ex} & 0 &
               1\hspace{-1ex} & 0\hspace{-1ex} & 0 &
               0\hspace{-1ex} & 0\hspace{-1ex} & 1
               \\
               \textcolor{blue}{{\tiny 2}} & 
               0\hspace{-1ex} & 1\hspace{-1ex} & 0 & 
               1\hspace{-1ex} & 0\hspace{-1ex} & 1 &
               1 &
               0\hspace{-1ex} & 1\hspace{-1ex} & 0 &
               0\hspace{-1ex} & 1\hspace{-1ex} & 0 &
               0\hspace{-1ex} & 1\hspace{-1ex} & 0 &
               1\hspace{-1ex} & 0\hspace{-1ex} & 0 
               \\
               \textcolor{blue}{{\tiny 3}} & 
               0\hspace{-1ex} & 0\hspace{-1ex} & 1 & 
               0\hspace{-1ex} & 1\hspace{-1ex} & 1 &
               1 &
               0\hspace{-1ex} & 0\hspace{-1ex} & 1  &
               0\hspace{-1ex} & 0\hspace{-1ex} & 1  &
               0\hspace{-1ex} & 0\hspace{-1ex} & 1  &
               0\hspace{-1ex} & 1\hspace{-1ex} & 0 
               \\
             \end{block}
\end{blockarray}

\caption{$\boldsymbol{H}_{p=2}\in \mathbb{F}_{q}^{3\times 19}$: If $\mathbb{F}_{q}$ does have characteristic two (such as $GF(2)$), then $\boldsymbol{H}_{2}$ is an encoding matrix for the index coding instance $\mathcal{I}_{F}(2)$, and if $\mathbb{F}_{q}$ does have odd characteristic any characteristic other than characteristic two such as $GF(3)$), then $\boldsymbol{H}_{2}$ is an encoding matrix for the index coding instance $\mathcal{I}_{nF}(2)$.}
\label{fig:H-p-2}
\end{figure*}

\end{exmp}

\begin{exmp}
Figure \ref{fig:H-p-3} depicts the encoding matrix $\boldsymbol{H}_{p=3}\in \mathbb{F}_{q}^{4\times 33}$ which is optimal for the class $3$-Fano index coding instance $\mathcal{I}_{nF}(3)$ if field $\mathbb{F}_{q}$ does have characteristic three.
    \begin{figure*}
\centering
\begin{blockarray}{cccccccccccccccccccccccccccccccccc}
              & 
              \textcolor{blue}{{\tiny 1}}\hspace{-1ex} & 
              \textcolor{blue}{{\tiny 2}}\hspace{-1ex} & 
              \textcolor{blue}{{\tiny 3}}\hspace{-1ex} & 
              \textcolor{blue}{{\tiny 4}}\hspace{-1ex} & 
              \textcolor{blue}{{\tiny 5}}\hspace{-1ex} & 
              \textcolor{blue}{{\tiny 6}}\hspace{-1ex} &
              \textcolor{blue}{{\tiny 7}}\hspace{-1ex} &
              \textcolor{blue}{{\tiny 8}}\hspace{-1ex} & 
              \textcolor{blue}{{\tiny 9}}\hspace{-1ex} & 
              \textcolor{blue}{{\tiny 10}}\hspace{-1ex} &
              \textcolor{blue}{{\tiny 11}}\hspace{-1ex} &
              \textcolor{blue}{{\tiny 12}}\hspace{-1ex} &
              \textcolor{blue}{{\tiny 13}}\hspace{-1ex} & 
              \textcolor{blue}{{\tiny 14}}\hspace{-1ex} &
              \textcolor{blue}{{\tiny 15}}\hspace{-1ex} &
              \textcolor{blue}{{\tiny 16}}\hspace{-1ex} &
              \textcolor{blue}{{\tiny 17}}\hspace{-1ex} & 
              \textcolor{blue}{{\tiny 18}}\hspace{-1ex} &
              \textcolor{blue}{{\tiny 19}}\hspace{-1ex} &
              \textcolor{blue}{{\tiny 20}}\hspace{-1ex} &
              \textcolor{blue}{{\tiny 21}}\hspace{-1ex} & 
              \textcolor{blue}{{\tiny 22}}\hspace{-1ex} &
              \textcolor{blue}{{\tiny 23}}\hspace{-1ex} &
              \textcolor{blue}{{\tiny 24}}\hspace{-1ex} &
              \textcolor{blue}{{\tiny 25}}\hspace{-1ex} & 
              \textcolor{blue}{{\tiny 26}}\hspace{-1ex} &
              \textcolor{blue}{{\tiny 27}}\hspace{-1ex} &
              \textcolor{blue}{{\tiny 28}}\hspace{-1ex} &
              \textcolor{blue}{{\tiny 29}}\hspace{-1ex} &  
              \textcolor{blue}{{\tiny 30}}\hspace{-1ex} &
              \textcolor{blue}{{\tiny 31}}\hspace{-1ex} &
              \textcolor{blue}{{\tiny 32}}\hspace{-1ex} &
              \textcolor{blue}{{\tiny 33}} 
            \\
\begin{block}{c[cccc|cccc|c|cccc|cccc|cccc|cccc|cccc|cccc]}
               \textcolor{blue}{{\tiny 1}} & 
               1\hspace{-1ex} & 0\hspace{-1ex} & 0\hspace{-1ex} & 0 & 
               1\hspace{-1ex} & 1\hspace{-1ex} & 1\hspace{-1ex} & 0 & 
               1 & 
               1\hspace{-1ex} & 0\hspace{-1ex} & 0\hspace{-1ex} & 0 & 
               1\hspace{-1ex} & 0\hspace{-1ex} & 0\hspace{-1ex} & 0 &
               1\hspace{-1ex} & 0\hspace{-1ex} & 0\hspace{-1ex} & 0 &
               1\hspace{-1ex} & 0\hspace{-1ex} & 0\hspace{-1ex} & 0 &
               0\hspace{-1ex} & 0\hspace{-1ex} & 0\hspace{-1ex} & 1 &
               0\hspace{-1ex} & 0\hspace{-1ex} & 1\hspace{-1ex} & 0 
               \\
               \textcolor{blue}{{\tiny 2}} & 
               0\hspace{-1ex} & 1\hspace{-1ex} & 0\hspace{-1ex} & 0 & 
               1\hspace{-1ex} & 1\hspace{-1ex} & 0\hspace{-1ex} & 1 &
               1 &
               0\hspace{-1ex} & 1\hspace{-1ex} & 0\hspace{-1ex} & 0 &
               0\hspace{-1ex} & 1\hspace{-1ex} & 0\hspace{-1ex} & 0 &
               0\hspace{-1ex} & 1\hspace{-1ex} & 0\hspace{-1ex} & 0 &
               0\hspace{-1ex} & 1\hspace{-1ex} & 0\hspace{-1ex} & 0 &
               1\hspace{-1ex} & 0\hspace{-1ex} & 0\hspace{-1ex} & 0 &
               0\hspace{-1ex} & 0\hspace{-1ex} & 0\hspace{-1ex} & 1 
               \\
               \textcolor{blue}{{\tiny 3}} & 
               0\hspace{-1ex} & 0\hspace{-1ex} & 1\hspace{-1ex} & 0 & 
               1\hspace{-1ex} & 0\hspace{-1ex} & 1\hspace{-1ex} & 1 &
               1 &
               0\hspace{-1ex} & 0\hspace{-1ex} & 1\hspace{-1ex} & 0  &
               0\hspace{-1ex} & 0\hspace{-1ex} & 1\hspace{-1ex} & 0  &
               0\hspace{-1ex} & 0\hspace{-1ex} & 1\hspace{-1ex} & 0 &
               0\hspace{-1ex} & 0\hspace{-1ex} & 1\hspace{-1ex} & 0 &
               0\hspace{-1ex} & 1\hspace{-1ex} & 0\hspace{-1ex} & 0 &
               1\hspace{-1ex} & 0\hspace{-1ex} & 0\hspace{-1ex} & 0 
               \\
               \textcolor{blue}{{\tiny 4}} & 
               0\hspace{-1ex} & 0\hspace{-1ex} & 0\hspace{-1ex} & 1 & 
               0\hspace{-1ex} & 1\hspace{-1ex} & 1\hspace{-1ex} & 1 & 
               1 &
               0\hspace{-1ex} & 0\hspace{-1ex} & 0\hspace{-1ex} & 1 &
               0\hspace{-1ex} & 0\hspace{-1ex} & 0\hspace{-1ex} & 1 &
               0\hspace{-1ex} & 0\hspace{-1ex} & 0\hspace{-1ex} & 1 &
               0\hspace{-1ex} & 0\hspace{-1ex} & 0\hspace{-1ex} & 1 &
               0\hspace{-1ex} & 0\hspace{-1ex} & 1\hspace{-1ex} & 0 &
               0\hspace{-1ex} & 1\hspace{-1ex} & 0\hspace{-1ex} & 0 
               \\
             \end{block}
\end{blockarray}

\caption{$\boldsymbol{H}_{p=3}\in \mathbb{F}_{q}^{4\times 33}$: If $\mathbb{F}_{q}$ does have characteristic three (such as $GF(3)$), then $\boldsymbol{H}_{3}$ is an encoding matrix for the index coding instance $\mathcal{I}_{F}(3)$, and if $\mathbb{F}_{q}$ does have any characteristic other than characteristic three (such as $GF(2)$), then $\boldsymbol{H}_{3}$ is an encoding matrix for the index coding instance $\mathcal{I}_{nF}(3)$.}
\label{fig:H-p-3}
\end{figure*}

\end{exmp}

\begin{prop} \label{prop:I-fano-necessary}
Matrix $\boldsymbol{H}\in \mathbb{F}_{q}^{(p+1)t\times mt}$ is an encoding matrix for index coding instance $\mathcal{I}_{F}(p)$ only if its submatrix $\boldsymbol{H}^{[p+1]}$ is a linear representation of matroid instance $\mathcal{N}_{F}(p)$.
\end{prop}

\begin{proof}
We prove $N_{0}=[p+1]$ is a basis set of $\boldsymbol{H}$, and each $N_{i}, i\in [p+1]$ in \eqref{eq:Fano-circuit-N_i} is a circuit set of $\boldsymbol{H}$. The proof is described as follows.

\begin{itemize}
    \item First, since $\beta_{\text{MAIS}}(\mathcal{I}_{F}(p)) = p+1$, we must have $\mathrm{rank}(\boldsymbol{H})=(p+1)t$. Now, from $B_{i}, i\in [p+1]$ in \eqref{eq:Fano-instane-b1n}, it can be seen that set $[p+1]$ is an independent set of $\mathcal{I}_{F}(p)$, so based on Lemma \ref{lem:MAIS1}, set $[p+1]$ is an independent set of $\boldsymbol{H}$. Since $\mathrm{rank}(\boldsymbol{H})=(p+1)t$, set $N_{0}=[p+1]$ will be a basis set of $\boldsymbol{H}$. Now, in order to have $\mathrm{rank} (\boldsymbol{H})=(p+1)t$ for all $j\in [m]\backslash [p+1]$, we must have $\mathrm{col}(\boldsymbol{H}^{\{j\}})\subseteq \mathrm{col}(\boldsymbol{H}^{[p+1]})$.
    \item According to Lemma \ref{lem:MAIS2}, from $B_{i}, i\in [p+1]$, it can be seen that for each $z(l,j), l\in [p+1], j\in [p+1]$, we have
    \begin{align}
        z(l,j)\in B_{i}, \forall i\in [p+1]\backslash \{j\}\rightarrow \boldsymbol{H}^{\{z(l,j)\}}=\boldsymbol{H}^{\{j\}}. 
        \label{eq:z(l,j)=j}
    \end{align}
    Thus,
    \begin{align}
        \boldsymbol{H}^{Z_{l}\backslash \{z(l,j)\}}=\boldsymbol{H}^{[p+1]\backslash \{j\}}, \ \forall j,l\in [p+1],
        \label{eq:I-fano-to-N-fano-2}
    \end{align}
    which means that each set $Z_{l}\backslash \{z(l,j)\}$ is an independent set of $\boldsymbol{H}$.
    \item To have $\mathrm{rank}(\boldsymbol{H})=(p+1)t$, we must have $\mathrm{rank}(\boldsymbol{H}^{B_{i}})=pt, i\in [m]$. Since $[p+1]$ is a basis set, from $B_{j}, j\in [p+1]$, one must have
    \begin{align}
        \mathrm{col}(\boldsymbol{H}^{\{n-j\}})\subseteq \mathrm{col}(\boldsymbol{H}^{[p+1]\backslash \{j\}})\stackrel{\eqref{eq:I-fano-to-N-fano-2}}{=}\mathrm{col}(\boldsymbol{H}^{Z_{l}\backslash \{z(l,j)\}}),
        \label{eq:I-fano-to-N-fano-3}
    \end{align}
    
    \item From $B_{i}, i\in Z_{l}, l\in [p+1]$, it can be checked that 
    each set $Z_{l}\backslash \{z(l,l)\}$ is a minimal cyclic set of $\mathcal{I}_{F}(p)$. Moreover, for each $l\in [p+1]$, we have $n-l \in B_{z(l,j)}, \forall j\in [p+1]$.
    
    \item Now, it can be seen that all four conditions in Lemma \ref{lem:independent-cycle} are met for each set $M=Z_{l}\backslash \{z(l,l)\}$ and $j=n-l$. Thus, based on Lemma \ref{lem:independent-cycle}, each set $(Z_{l}\backslash \{z(l,l)\})\cup \{n-l\}, l\in [p+1]$ must be a circuit set of $\boldsymbol{H}$.
    
    \item Let $M_{1}= Z^{\prime}\cup Z^{\prime\prime}$ and $M_{2}= \{z(l,l), n-l, n\}$. It can be seen that set $M_{1}$ is an acyclic set of $\mathcal{I}_{F}(p)$ and $M_{2}\subset B_{j}$ for all $j\in M_{1}$.  According to Lemma \ref{lem:i-n-i-n}, we must have
    \begin{align}
        \mathrm{rank}(\boldsymbol{H}^{\{z(l,l), n-l, n\}})\leq 2t,
    \end{align}
    which according to \eqref{eq:z(l,j)=j} will lead to
        \begin{align}
        \mathrm{rank}(\boldsymbol{H}^{\{l, n-l, n\}})\leq 2t.
    \end{align}
    Thus, each set $\{l, n-l, n\}, l\in [p+1]$ is a circuit set of $\boldsymbol{H}$.

    \item Finally, from $B_{n-j}, j\in [p+1]$, it can be observed that set $\{(n-j), j\in [p+1]\}$ is a minimal cyclic set of $\mathcal{I}_{F}(p)$. Furthermore, from $B_{n}$, one must have $\mathrm{rank}( \boldsymbol{H}^{B_{n}})=\mathrm{rank}( \boldsymbol{H}^{\{(n-j), j\in [p+1]\}})=pt$. Thus, based on Lemma \ref{lem:MAIS2}, set $\{(n-j), j\in [p+1]\}$ must be a circuit set of $\boldsymbol{H}$. This completes the proof.

\end{itemize}
\end{proof}

\subsection{The Class $p$-non-Fano Index Coding Instance} \label{sub:index-coding-non-fano}

\begin{defn}[$\mathcal{I}_{nF}(p)$: Class $p$-non-Fano Index Coding Instance] For the class of index coding instances $\mathcal{I}_{nF}(p)=\{B_{i}, i\in [m]\}$ with $m=2p^{2}+4p+3$, the interfering message sets $B_{i}, i\in [m]\backslash (\{n-j, j\in [p+1]\} \cup \{n\}), n=2p+3$ are exactly the same as the ones in $\mathcal{I}_{F}(p)$, and the interfering message sets $B_{n-j}, j\in [p+1]$ are as follows

\begin{align}
    &B_{n-j}=\{n-l, l\in [p+1]\}\backslash \{j\}, j\in [p+1],
    \nonumber
    \\
    &B_{n}=\emptyset.
    \label{eq:Bi-non-Fano-[p+2:2p+2]}
\end{align}
\end{defn}

\begin{thm}
$\lambda_{q}(\mathcal{I}_{nF}(p))=\beta (\mathcal{I}_{nF}(p))=p+1$ iff the field $\mathbb{F}_{q}$ does have any characteristic other than characteristic $p$. In other words, the necessary and sufficient condition for linear index coding to be optimal for $\mathcal{I}_{nF}(p)$ is that the field $\mathbb{F}_{q}$ has any characteristic other than characteristic $p$.
\end{thm}
\begin{proof}
The proof can be concluded from Propositions  \ref{prop:I-non-Fano-necessary-index} and \ref{prop:I-non-fano-necessary}.
\end{proof}

\begin{prop} \label{prop:I-non-Fano-necessary-index}
For the class $p$-non-Fano index coding instance, there exists an optimal scalar linear code $(t=1)$ over the field with any characteristic other than characteristic $p$.
\end{prop}

\begin{proof}
We show that the encoder matrix shown in Figure \ref{fig:H-class-fano-nonfano} can satisfy all the users of $\mathcal{I}_{nF}(p)$. Since the interfering message sets $B_{i}, i\in [m]\backslash [p+2:2p+3]$ are the same as the ones in $\mathcal{I}_{F}(p)$, they will be satisfied by the encoder matrix shown in Figure \ref{fig:H-class-fano-nonfano}. Since for each user $u_{i}, i\in [p+2:2p+2]$,
\begin{align}
    \boldsymbol{H}^{\{i\}\cup B_{i}}=\boldsymbol{H}^{[p+2:2p+2]},
\end{align}
and according to Lemma \ref{lem:matrix-F-(p+1)-p}, the submatrix $\boldsymbol{H}^{[p+2:2p+2]}$ is full-rank and invertible over the fields with any characteristic other than characteristic $p$, the column $\boldsymbol{H}^{\{i\}}$ is linearly independent of the space spanned by the columns of $\boldsymbol{H}^{B_{i}}$. Thus, the decoding condition in \eqref{eq:dec-cond} will be satisfied for all users $u_{i}, i\in [p+2:2p+2]$. Moreover, it can be easily seen that the decoding condition will be satisfied for user $u_{2p+3}$, as $B_{2p+3}=\emptyset$. This completes the proof.
\end{proof}

\begin{prop} \label{prop:I-non-fano-necessary}
Matrix $\boldsymbol{H}\in \mathbb{F}_{q}^{(p+1)t\times mt}$ is an encoding matrix for index coding instance $\mathcal{I}_{nF}(p)$ only if its submatrix $\boldsymbol{H}^{[p+1]}$ is a linear representation of matroid instance $\mathcal{N}_{nF}(p)$.
\end{prop}

\begin{proof}
Since the interfering message sets $B_{i}, i\in [m]\backslash (\{n-j, j\in [p+1]\}\cup \{n\}$ are the same as the ones in $\mathcal{I}_{F}(p)$, it can be proved that set $N_{0}=[p+1]$ is a basis set of $\boldsymbol{H}$ and each set $N_{i}, i\in [p+1]$ in \eqref{eq:Fano-circuit-N_i} is a circuit set of $\boldsymbol{H}$. Now, from \eqref{eq:Bi-non-Fano-[p+2:2p+2]}, it can be seen that set $\{n-j, j\in [p+1]\}$ is an independent set of $\mathcal{I}_{nF}(p)$. Thus, according to Lemma \ref{lem:MAIS1}, set $N_{n}=\{n-j, j\in [p+1]\}=[p+2:2p+2]$ must be an independent set of $\boldsymbol{H}$. This completes the proof.
\end{proof}

\section{Conclusion} \label{sec:conclusion}

Matroid theory is intrinsically linked to index coding and network coding problems. Indeed, the dependency of linear index coding and network coding rates on the characteristic of a field has been illustrated through two notable matroid instances, the Fano and non-Fano matroids. This relationship highlights the limitations of linear coding, a key theorem in both index coding and network coding domains. Specifically, the Fano matroid can only be linearly represented in fields with characteristic two, while the non-Fano matroid is linearly representable exclusively in fields with odd characteristics.
For fields with any characteristic $p$, the Fano and non-Fano matroids have been extended into new classes whose linear representations rely on fields with characteristic $p$. Despite this extension, these matroids have not received significant recognition within the fields of network coding and index coding.
In this paper, we first presented these matroids in a more accessible manner. Subsequently, we proposed an independent alternative proof, primarily utilizing matrix manipulation instead of delving into the intricate details of number theory and matroid theory. This novel proof demonstrates that while the class $p$-Fano matroid instances can only be linearly represented in fields with characteristic $p$, the class $p$-non-Fano instances are representable in fields with any characteristic except $p$.
Finally, following the properties of the class $p$-Fano and $p$-non-Fano matroid instances, we characterized two new classes of index coding instances, namely the class $p$-Fano and $p$-non-Fano index coding, each comprising $p^2 + 4p + 3$ users. We demonstrated that the broadcast rate for the class $p$-Fano index coding instances can be achieved through linear coding only in fields with characteristic $p$. In contrast, for the class $p$-non-Fano index coding instances, it was proved that the broadcast rate can be achieved through linear coding in fields with any characteristic except $p$.

\IEEEpeerreviewmaketitle
\bibliographystyle{IEEEtran}
	\bibliography{main}

\end{document}